\documentclass[runningheads]{llncs}
\usepackage[margin=1.2in]{geometry}
\usepackage[numbers,sort]{natbib}

\usepackage{amssymb}
\usepackage{subcaption}

\usepackage{amsthm}

\usepackage{auxlib}

\begin{document}

\title{Online Selfish Load Balancing}
\author{
Wenqian Wang\inst{1}
\and
Chenyang Xu\inst{2}
\and
Yuhao Zhang\inst{1}
}

\institute{
School of Computer Science, Shanghai Jiao Tong University, China
\and
School of Computer Science and Technology, East China Normal University, China
}
\maketitle
\begin{abstract}

In selfish load balancing, there is a set of machines and jobs to be scheduled, where each machine is owned by a selfish agent. Agents hold the processing times of jobs as private information and may strategically misreport them to maximize their utilities. The goal is to design a truthful mechanism that minimizes the makespan. This selfish-machine model was first proposed by Nisan and Ronen (STOC~1999), who presented an $m$-approximation algorithm for unrelated machines in the offline scenario, which was later shown to be tight by Christodoulou, Koutsoupias, and Kovács (STOC~2023). The study of offline selfish load balancing on related machines was initiated by Archer and Tardos (FOCS~2001). The best-known results for this problem are two PTAS mechanisms, due to Christodoulou and Kovács (SICOMP~2013) and Epstein, Levin and Stee (MOR~2016).

However, there is little literature on selfish scheduling in the online scenario, which is precisely what arises in real-world applications (e.g., in cloud platforms, jobs often arrive online). In this paper, we aim to address this gap. For unrelated machines, we observe that the existing $m$-approximation algorithm can also be implemented in an online scenario, implying that $m$ remains the best possible competitive ratio.  
For related machines, we design the first nontrivial online mechanism that is truthful in expectation and achieves a competitive ratio of~$O(\log m)$. Moreover, we extend our mechanism to also guarantee job-side truthfulness (in expectation), ensuring that jobs arriving online report their true sizes. This notion was first studied by Feldman, Fiat, and Roytman (EC~2017), but without combining it with the classic machine-side truthfulness. Finally, we generalize our two-sided truthful mechanism to the $\ell_q$-norm variant of load balancing, achieving a competitive ratio of~$\tilde{O}\left(m^{\frac{1}{q}(1-\frac{1}{q})}\right)$.

\end{abstract}



\section{Introduction}



Selfish load balancing is one of the most classic problems in algorithmic mechanism design. In this model, there is a set of machines and a set of jobs to be scheduled on them. Each machine is owned by a selfish agent, who knows her own processing time for each job (as private information) and seeks to maximize her utility, defined as the payment minus the completion time. The goal is to design a truthful mechanism that minimizes the makespan.
The model was first introduced by~\citet{DBLP:conf/stoc/NisanR99}, who provided an $m$-approximation mechanism and conjectured that no truthful mechanism could achieve a better approximation ratio than $m$. This conjecture inspired a large body of follow-up work~\cite{DBLP:journals/algorithmica/ChristodoulouKV09,DBLP:journals/algorithmica/KoutsoupiasV13,DBLP:conf/focs/0001KK21}, and it was finally confirmed more than two decades later by~\citet{DBLP:conf/stoc/0001KK23}. 

Shortly after the original 1999 model was proposed, it was extended to the \emph{related machines} setting~\cite{DBLP:conf/focs/ArcherT01}, where each machine has her speed as private information. The authors presented a $3$-approximation randomized algorithm that is truthful (in expectation) on the machine side. Many subsequent studies followed~\cite{archer2004mechanisms,DBLP:journals/siamcomp/DhangwatnotaiDDR11,DBLP:conf/stacs/AndelmanAS05,DBLP:conf/esa/Kovacs05,DBLP:journals/jda/Kovacs09}, culminating in two truthful deterministic PTAS mechanisms proposed by~\citet{DBLP:journals/siamcomp/0001K13} and~\citet{DBLP:journals/mor/EpsteinLS16}.

Selfish load balancing captures many real-world scenarios, such as cloud computing, crowdsourcing platforms, and large-scale manufacturing systems, where machines are owned by strategic agents and jobs need to be assigned efficiently. However, in the aforementioned scheduling models, all jobs are assumed to be known in advance, whereas in reality, such scenarios often exhibit an \emph{online} nature, with jobs arriving sequentially. For example, in cloud computing, servers must report their configurations to the platform in advance, even though the jobs to be executed have not yet arrived. In other words, agents must make decisions without knowing which jobs will arrive in the future. 

Despite its practical relevance, there has been limited research on the online selfish setting. For unrelated machines, we observe that the $m$-approximation mechanism proposed by~\citet{DBLP:conf/stoc/NisanR99} still applies in the online case\footnote{
~\cref{sec:unrelated} 
for more details.}, as each job’s assignment in this mechanism is independent of the others. Therefore, the machines’ strategies do not depend on knowledge of future jobs: upon the arrival of each job, each machine simply reports its corresponding processing time truthfully.
For related machines, however, the previously developed constant-approximation mechanisms no longer apply in the online setting, since they critically rely on global information about all jobs to achieve both truthfulness and a good approximation ratio. \citet{DBLP:journals/tcs/AulettaPPP09} considered the special case of $m=2$ and proposed a mechanism achieving an $(1+\sqrt{7}/2)$-approximation. For the general case of arbitrary~$m$, to the best of our knowledge, the competitive ratio of~$m$ still remains the best known upper bound.  
Thus, the following open question naturally arises:

\vspace{2mm}
\begin{minipage}{0.9\linewidth}
\begin{center}
   \emph{Does there exist a truthful online mechanism for scheduling on related machines that achieves a competitive ratio strictly better than $m$?} 
\end{center}
\end{minipage}
\vspace{2mm}



Towards this open question, one useful result~\cite{DBLP:journals/tcs/AulettaPPP09} is an equivalent characterization of truthfulness, known as \emph{load monotonicity}: the total job size (load) assigned to a machine must be non-decreasing in its reported speed. An attempt to address this open question was made by~\citet{DBLP:journals/scheduling/LiLW23}, who focused on \emph{well-behaved} mechanism design, where a machine’s load is always at least as large as the load on any slower machine.  
The well-behavior property, first proposed by~\citet{DBLP:journals/mor/EpsteinLS16}, shares a similar intuition with load monotonicity—both aim to assign more load to faster machines.~\citet{DBLP:journals/scheduling/LiLW23} proposed an (almost)\footnote{The term ``almost well-behaved'' means that the allocation violates the well-behavior property only between two machines whose speeds differ by no more than a factor of two.} well-behaved mechanism that achieves an $O(\log m)$ competitive ratio. Due to this similar intuition, they hope their algorithm and analysis can shed light on the design of truthful mechanisms. As supporting evidence, they show that their mechanism is indeed truthful in special cases where the number of machines is two or where the jobs are identical.


\subsection{Our Contributions}

In this paper, we make progress on this open question and give the \emph{first non-trivial} truthful mechanism:

\paragraph{Main Result.} For online selfish load balancing, there exists a randomized polynomial mechanism which is truthful in expectation and yields a competitive ratio of $O(\log m)$ both in expectation and with high probability.
\vspace{2mm}

Addressing the online selfish setting presents significant technical challenges, as we observe that many intuitive and reasonable mechanisms fail to achieve online truthfulness.  
For instance, the classical online greedy algorithm~\cite{DBLP:journals/jacm/AspnesAFPW97} fails to maintain machine-side truthfulness; on the other hand, offline machine-side truthful mechanisms, such as~\cite{DBLP:conf/focs/ArcherT01,DBLP:conf/esa/Kovacs05,DBLP:journals/mor/EpsteinLS16}, typically rely on offline operations like sorting, which are infeasible in the online setting.  
More importantly, we find that the recently proposed (almost) well-behaved algorithm~\cite{DBLP:journals/scheduling/LiLW23}, although truthful on the machine side in certain special cases (e.g., $m=2$), is not truthful in general.  
We illustrate a counterexample involving only three machines in  
~\cref{sec:ce}.

To achieve our main result, we combine proportional allocation with online truthful parametric pruning\footnote{parametric pruning means guess a parameter, then prune the instance according to that parameter. Here, in our model, we guess an OPT online and restrict the machine subset that a job can be assigned to. } to design a logarithmic-competitive truthful mechanism.  
We note that, following many prior works on offline selfish machine scheduling—such as the first truthful mechanism~\cite{DBLP:conf/focs/ArcherT01} and the first truthful PTAS mechanism~\cite{DBLP:journals/siamcomp/DhangwatnotaiDDR11}—we assume each agent is risk-neutral\footnote{Each agent cares only about the expected utility. For example, receiving a utility of~\(M\) with probability~\(1/M\) is equivalent to receiving a utility of~1 with probability~1 from the agent's perspective.} and employ a randomized mechanism to achieve truthfulness.  
Furthermore, we remark that our mechanism also preserves the well-behavior property in expectation.

\subsubsection*{Extension to Online Selfish-Job Scheduling.}  
In the selfish-job model, each job is controlled by a selfish agent, who holds the job size as private information and seeks to minimize the sum of its completion time and payment.
Handling selfish online jobs (assuming machines are not selfish and related) has been well studied in the literature.  
A natural approach to achieving truthfulness is the greedy algorithm~\cite{DBLP:journals/jacm/AspnesAFPW97}, which is \(O(\log m)\)-competitive even when machines have different speeds (i.e., related machines).~\citet{DBLP:conf/sigecom/FeldmanFR17} proposed an \(O(1)\)-competitive, job-side truthful mechanism for related machines using a dynamic posted pricing framework.  
Building on the selfish-machine setting, we further consider this selfish-job setting and establish the following.

\paragraph{Extended Result~\Rom{1}}  
The aforementioned mechanism can be extended to a two-sided truthful (in expectation) mechanism, incurring only a constant-factor loss in the corresponding competitive ratios.  
Here, two-sided means that the mechanism is truthful simultaneously for both selfish machines and selfish jobs.
\vspace{2mm}

To extend our mechanism to the selfish job setting, we first, similar to the machine side, prove an equivalent characterization of truthfulness, job-side monotonicity: the speed of the machine assigned to a job is non-decreasing as the reported job size increases.  
We observe that previous mechanisms violate job-side monotonicity due to the use of online parametric pruning.  
To address this issue, we introduce the \emph{allocate-before-doubling} trick.    
This trick eliminates the negative impact of online parametric pruning on job-side monotonicity, and we further show that it incurs only a constant-factor loss in the competitive ratio.

\subsubsection*{Extension to $\ell_q$-Norm Minimization.} We further consider online selfish scheduling with an $\ell_q$ norm objective. There exist some works for offline truthful $\ell_q$ norm minimization~\cite{DBLP:journals/mor/EpsteinLS16}, but their techniques are hard to extend to the online scenario. There is also a work of \citet{DBLP:conf/stoc/ImKPS18} studying online scheduling under $\ell_q$ norm, which establishes an $O(1)$-competitive online algorithm, but not aiming for truthfulness. In this model, when $q=1$, the optimal solution can be easily achieved by assigning each job to the fastest machine; however, as $q$ increases beyond 1, achieving a polylogarithmic competitive ratio becomes challenging. 
By extending our techniques on makespan minimization, we obtain the \emph{first non-trivial} result for the general $\ell_q$ norm objective in the online setting under selfish agents.

\paragraph{Extended Result \Rom{2}} For online selfish scheduling with $\ell_q$ norm minimization, there exists a randomized polynomial time mechanism which is two-sided truthful (in expectation) and yields a competitive ratio of $O(m^{\frac{1}{q}(1-\frac{1}{q})} \cdot (\log m)^{1+\frac{1}{q^2}} )$ both in expectation and with high probability.
\vspace{2mm}

The above ratio remains consistent with our makespan algorithm, which will approach $O(\log m)$ as $q$ approaches $\infty$. Our ratio is strictly better than the trivial competitive ratio $O(m^{1-1/q})$\footnote{Assign all jobs to the fastest machine.} if omitting the lower-order terms. 
Note that in~\citet{DBLP:conf/stoc/ImKPS18}'s $O(1)$-competitive algorithm, there is an involved preprocessing subroutine to round machines' speed, which is a very crucial technique to achieve the constant competitive ratio. However, this subroutine may introduce an unexpected change when one machine speeds up (e.g., after a machine speeds up, another machine's rounding speed may also change), implying that their algorithm is not monotone.
In our mechanism, we also round the speeds of machines, but only in a simple way --- rounding down to the nearest power of~$2$. It may introduce a loss of competitive ratio, but it does not affect the monotonicity.

\subsection{Paper Organization}

This paper presents an online mechanism with an $O(\log m)$ competitive ratio that guarantees not only machine-side truthfulness but also job-side truthfulness.  
We begin by formally defining our model in \Cref{sec:pre}, introducing the notion of truthfulness and its reduction to monotonicity.  
In \Cref{sec:tech}, we provide an overview of our algorithmic approach and highlight the key technical ideas.~\Cref{sec:makespan} presents our truthful mechanism for makespan minimization along with its analysis.  
In particular, we prove both our Main Result and Extended Result~\Rom{1} in this section (see \Cref{thm:makespan}).  
We further show that both the algorithmic framework and the analysis naturally extend to the $\ell_q$-norm setting, where we prove Extended Result~\Rom{2} in \Cref{sec:norm} (see \Cref{thm:lq}).

\subsection{Other Related Work}

The offline scheduling problem is a well-studied model with extensive research results. For identical or related machines, there are PTAS results~\cite{DBLP:journals/jal/KargerPT96,DBLP:journals/jal/BermanCK00,DBLP:conf/soda/AlonAWY97,DBLP:journals/algorithmica/EpsteinS04} with makespan and $\ell_q$-norm minimization. For unrelated machines, \citet{DBLP:journals/mp/LenstraST90} gives a $2$-approximation algorithm.

For online load balancing, the model is initiated by~\citet{DBLP:journals/siamam/Graham69}, and subsequently, numerous studies have extensively explored various settings. For identical machines, \citet{DBLP:journals/siamam/Graham69} gives a deterministic algorithm that achieves a ratio of $2 - 1/m$.
Many subsequent works have aimed to improve this ratio~\cite{DBLP:journals/siamcomp/Albers99,DBLP:journals/jcss/BartalFKV95,DBLP:journals/jal/KargerPT96} and have also investigated the hardness of the problem~\cite{DBLP:journals/siamcomp/Albers99,DBLP:journals/ipl/BartalKR94,DBLP:journals/actaC/FaigleKT89,DBLP:conf/soda/GormleyRTW00}. The current best results for deterministic algorithms are an upper bound of $1.9201$~\cite{DBLP:conf/esa/FleischerW00} and a lower bound of $1.880$~\cite{DBLP:journals/siamcomp/RudinC03}.
There are also many works focusing on randomized algorithm design~\cite{DBLP:journals/algorithmica/Seiden00,DBLP:journals/jcss/BartalFKV95,DBLP:journals/disopt/DosaT10,DBLP:journals/orl/Tichy04}.
The current best results are an upper bound of $1.916$~\cite{DBLP:conf/stoc/Albers02} and a lower bound of $1.582$~\cite{DBLP:journals/ipl/Sgall97}. 

For related machines, the deterministic slow-fit algorithm proposed by~\citet{DBLP:conf/stoc/AspnesAFPW93} achieves a competitive ratio of $8$. The optimal competitive ratio of deterministic algorithms lies between $2.56$~\cite{DBLP:journals/mst/EbenlendrS15} and $5.82$~\cite{DBLP:journals/jal/BermanCK00} and that of randomized algorithms is between $2$~\cite{DBLP:journals/orl/EpsteinS00} and $4.311$~\cite{DBLP:journals/jal/BermanCK00}. 

For unrelated machines, the optimal competitive ratio is proved to be $\Theta(\log m)$~\cite{DBLP:journals/jacm/AspnesAFPW97,DBLP:journals/jal/AzarNR95}. Some works also focus on $\ell_q$ norm, providing a $2-\Theta(\ln{q}/q) $-competitive~\cite{DBLP:journals/algorithmica/AvidorAS01} upper bound for identical machines, an $O(1)$-competitive~\cite{DBLP:conf/stoc/ImKPS18} algorithm for related machines, and a tight $\Theta(q)$-competitive~\cite{DBLP:conf/focs/AwerbuchAGKKV95,DBLP:conf/soda/Caragiannis08} algorithm for unrelated machines.

\section{Preliminaries}\label{sec:pre}

In this section, we formally describe the model and summarize the results established in the literature.  
To avoid redundancy, we present a general model that incorporates both selfish machines and selfish jobs.

\paragraph{Online Two-Sided Selfish Scheduling.} There are $m$ selfish machines $\cM$ with private speeds $\vs = \{s_i\}_{i\in [m]}$ and $n$ selfish jobs $\cJ$ with private sizes $\vp = \{p_j\}_{j\in [n]}$. Initially, each machine reports a speed, and then jobs arrive one by one. Each time a job $j$ arrives, it reports a size, and we must immediately assign it to a machine and charge a fee $Q_j$; after all jobs have been processed, we pay the machines some money. Each machine $i\in [m]$ follows the classic quasi-linear utility $P_i - C_i$, where $P_i$ is its payment, $C_i = L_i/s_i$ is the total completion time, and $L_i$ is the total assigned job size; while the objective of each job is to minimize its completion time (denoted by $F_j$) plus the money it pays ($Q_j$). It is worth noting that when a job arrives, it can access the speed and current load of each machine (if previous jobs were assigned randomly, it can also observe the random assignment results). 
Let $\vC := \langle C_1,...,C_m \rangle$ be the vector of machine completion times. Our goal is to design a mechanism that is truthful on both the job side and the machine side to minimize the $\ell_q$ norm $|| \vC ||_q$. When $q=\infty$, the objective is the same as the makespan, i.e., the maximum completion time. 

We also consider the fractional relaxation of the problem, which is used as an intermediate step for a randomized algorithm. In the fractional version, we can have a fractional assignment for each job $j$. We use $x_{ij}$ to denote the fraction of $j$ scheduled on $i$. We have $\sum_{i=1}^{m} x_{ij} = 1$. The load contributed to machine $i$ by job $j$ is $p_j \cdot x_{ij}$, the corresponding processing time is $p_j \cdot x_{ij} / s_i$.

\paragraph{Machine-Side and Job-Side Truthfulness.}
By standard definition, a mechanism is considered truthful if reporting true information is always the best strategy for each selfish agent.
Specifically, we call a mechanism \emph{machine-side truthful} if, for every machine $i$,  reporting its real speed $s_i$ can maximize its (expected) utility $\E[P_i-C_i]$, even after he knows all online jobs at the end. We call a mechanism \emph{job-side truthful} if, for every job $j$, after knowing the current load of every machine at its arrival, reporting its real size $p_j$, maximizes its expected utility $-\E[F_j + Q_j]$.
If a mechanism is both machine-side truthful and job-side truthful, we call it \emph{two-sided truthful}. 

We aim to design a centralized (randomized) online algorithm that can be implemented by a two-sided truthful mechanism with two specific payment functions on the machine side and the job side. If it can be done, we call such algorithms machine-side implementable and job-side implementable, respectively. If both, then they are two-sided implementable. The following lemma is well-known in the literature.

\begin{definition}[Machine-Side Monotone]
   A (randomized) online algorithm is said to be \emph{machine-side monotone} if, for each machine, unilaterally increasing its speed results in a non-decreasing (expected) load scheduled on that machine.
\end{definition}

\begin{lemma}[\citet{DBLP:conf/focs/ArcherT01}]
\label{lem:machine-side-implementable}
    A (randomized) online algorithm is machine-side implementable if and only if it is machine-side monotone. 
\end{lemma}

\citet{DBLP:conf/focs/ArcherT01} originally proposed the equivalence lemma in the context of offline mechanisms. However, as noted by \citet{DBLP:journals/scheduling/LiLW23}, the proof remains valid when applied to online algorithms. On the other hand, inspired by this equivalence lemma, we establish a symmetric version on the job side using a similar proof. 
For completeness, we present the proof in 
\Cref{sec:job-truth}.


\begin{definition}[Job-Side Monotone]
A (randomized) online algorithm is \emph{job-side monotone} if, for each job $j$, the (expected) unit processing time $ \sum_{i\in [m]} x_{ij}\cdot \frac{1}{s_i}$ is non-increasing as $p_j$ increases unilaterally.
\end{definition}

\begin{lemma}\label{lem:job-side-implementable}
A (randomized) online algorithm is job-side implementable if and only if it is job-side monotone.
\end{lemma}
    

\section{Overview of Our Techniques}
\label{sec:tech}



In this section, we summarize our techniques in mechanism design and analysis. 

\subsection{Algorithmic Framework}

Based on~\cref{lem:machine-side-implementable} and~\cref{lem:job-side-implementable}, we can reduce the task of two-sided truthful mechanism design to designing an online algorithm that is both machine-side and job-side monotone.
Instead of directly designing an integral allocation algorithm,
we first approach the problem fractionally and then employ a randomized rounding algorithm to round the fractional solution. We aim to mitigate the allocation instability caused by machine speed-ups and maintain machine-side monotonicity.


To obtain a desirable competitive ratio, we require that the online algorithm not only exhibits the competitiveness of the returned fractional solution but also satisfies a crucial property called \emph{speed-size-feasible}. This property ensures that the fractional competitive ratio can be approximately maintained after independent rounding. The definition of speed-size-feasibility is as follows.


\begin{definition}[Speed-Size-Feasible]\label{def:speed_constraints} 
   A fractional allocation $\vX=\{x_{ij}\}_{i\in [m],j\in [n]}$ is \emph{speed-size-feasible} if for any machine $i$ and any job $j$ with $x_{ij}>0$, we have 
   \begin{enumerate}
       \item $p_{j} \leq c \cdot s_i \cdot \opt$, where $c$ is a constant. 
       \item $s_i \geq \frac{s_1}{c \cdot m}$, where $c$ is a constant and $s_1$ is the highest machine speed. 
   \end{enumerate}
\end{definition}

We show that for any speed-size feasible solution that is $\alpha$-competitive,
independent rounding yields a competitive ratio of $O(\max\{\log m, \alpha\})$ both in expectation and with high probability (\Cref{lem:rounding}). Roughly speaking, the rounding incurs a ratio loss bounded by $O(\log m)$.
Following this framework, our primary objective shifts to finding a speed-size-feasible and two-sided monotone allocation that achieves a comparable competitive ratio.


\subsection{Proportional Allocation with Truthful Parametric Pruning}

This subsection illustrates the intuition of how to design a speed-size-feasible fractional algorithm with two-sided monotonicity, especially the machine-side monotonicity. 
We observe that when an algorithm's allocation of a specific job heavily depends on the allocations of previous jobs, machine-side monotonicity can easily be violated, as illustrated by the counterexample of \citet{DBLP:journals/scheduling/LiLW23}'s algorithm in 
\Cref{sec:ce}.

Roughly speaking, when a machine speeds up, a machine-side monotone algorithm should have an incentive to schedule more jobs on it. However, if the allocation of later jobs heavily depends on the allocations of earlier jobs, the speed-up machine may initially receive a heavier load than before, but ultimately end up with less load due to unexpected shifts in job assignments.
To address this, we aim to design an algorithm that allocates each job independently. This approach ensures that when a machine speeds up, the allocation of previous jobs does not inadvertently affect the allocation of subsequent jobs, allowing the speed-up machine to maintain its advantage across all jobs.

A very natural way to implement this approach is through a proportional allocation scheme, which assigns each arriving job to machines in proportion to their respective speeds. This ensures that each job's allocation is entirely independent, contributing to the algorithm's stability. However, blindly distributing jobs in this manner can violate speed-size feasibility, resulting in a poor competitive ratio after rounding. \footnote{We can prove every machine's completion time is good in expectation, but the makespan, i.e., the max of them, is not good, even in expectation.}



To satisfy the speed-size-feasible property while maintaining a good competitive ratio, online parametric pruning is essential. This means we cannot allocate any fraction of a job to a machine that is too slow. In the traditional online setting, many algorithms, such as slow-fit \cite{DBLP:journals/jal/BermanCK00} embed a doubling procedure to maintain a guess of \opt, and subsequently identify and eliminate machines that are considered excessively slow for each job. However, this doubling procedure introduces several unexpected issues regarding machine-side monotonicity. (See 
\Cref{sec:ce} 
for a counter-example of a seemingly monotone algorithm called water-filling, which fails to be monotone with the doubling procedure.) 
For example, when we implement proportional allocation, the presence of the doubling procedure breaks the independence of job allocations, leading to instability in the algorithm. Consequently, designing the doubling operation becomes a critical challenge.

To address it, we use two critical ideas in doubling. 
The first is the \emph{double-without-the-last} technique. 
In our mechanism, we will partition jobs into different levels according to their sizes. We will omit the jobs of the last level (small jobs) in the doubling criteria. This approach will keep the doubling operation stable, which preserves the machine-side monotonicity.
The second is \emph{allocate-before-doubling}, which aims to prevent the manipulation by individual jobs during the doubling process, thereby preserving the job-side monotonicity.

\subsection{Generalization to $\ell_q$ Norm}

In the next step, we generalize our approach to $\ell_q$ norm. The framework is the same as before, so we also aim for a proportional algorithm that is speed-size-feasible and two-sided monotone. The main difference is that the proportional assignment and the doubling condition are both adjusted to create a generalized version that fits the $\ell_q$ norm and is consistent with the makespan when $q=\infty$. As a result, the proof of two-sided monotonicity can be directly generalized. For the competitive ratio, the speed-size-feasible property cannot guarantee an $O(\log m)$ loss. In fact, the rounding may introduce an $O(\log m)$ loss on each machine's load. This implies an $O(\log m)$ loss on the competitive ratio for makespan but a $\tilde{O}(m^{1/q})$ loss for the $\ell_q$ norm.

To achieve a better analysis, we partition the machines into several groups and evaluate each group's rounding performance together, thus achieving a competitive ratio of $\tilde{O}(m^{\frac{1}{q}(1-\frac{1}{q})})$. 

\section{Two-Sided Truthful Mechanism for Makespan}
\label{sec:makespan}


In this section, we consider makespan minimization for the selfish scheduling model and aim to show the following.

\begin{theorem}\label{thm:makespan}
    For online two-sided selfish scheduling with makespan minimization, there exists a randomized polynomial mechanism that is two-sided truthful and yields a competitive ratio of $O(\log m)$ both in expectation and with high probability.
\end{theorem}

Following our framework, it is sufficient to design a speed-size-feasible and two-sided monotone fractional algorithm to prove the theorem above. We present our main algorithm, which is called \emph{level-based proportional allocation}, as follows.

The algorithm first normalizes the machine speeds. We round down each $s_i$ to the nearest power of 2, denoted by $\bs_i$. Then, we partition them into $K=\lfloor \log m \rfloor+1$ groups. The first group $\cM_1$ consists of the machines with speed $\bs_1$. The other levels, e.g., the $k$-th group, denoted by $\cM_k$, includes the machines with speed $r_k = \frac{\bs_1}{2^{k-1}}$. Notably, this means we omit all machines with speeds lower than $\frac{s_1}{m}$. Note that this normalization differs from the standard approach which uses $s_1$ as a reference and rounds the speed of group $k$ to $\frac{s_1}{2^{k-1}}$. That approach may mess up machine groupings when the first machine speeds up.
In contrast, we round $s_1$ to the nearest power of $2$ to maintain the stability of machine grouping.

Then, when the first job arrives, the algorithm distributes it among the machines in the first group and initiates a guessed \opt called $\Lambda=p_1/r_1$.  For any other arriving job $j$, we first calculate $j$'s size level, denoted by $k(j)$. In particular, $k(j)$ is the index of the slowest group that can accept job $j$ w.r.t. the current $\Lambda$ --- a group $k$ can accept job $j$ if $p_j \leq r_k \cdot \Lambda$. If a job cannot be accepted by any group, we will mark it as a super large job and set $k(j)=1$. 
After calculating job $j$'s level, we perform a fractional assignment for job $j$ to all the machines in $\cM_1, \cM_2, \dots \cM_{k(j)}$, proportional to the machines' speeds. Specifically, let $\cpM{k(j)} := \bigcup_{k'\leq k(j)} \cM_{k'}$ denote the feasible machine set for job $j$ and $x_{ij}$ represent the assigned fraction of job $j$ to machine $i$. The algorithm sets $x_{ij}$ to be $\frac{\bs_i}{\sum_{t\in \cpM{\kj}} \bs_{t}}$. This proportional assignment ensures that each machine $i\in \cpM{\kj}$ experiences an identical increase in completion time (i.e. $\frac{x_{ij}}{\bs_i}$).
Furthermore, we note that all the jobs at the same level will use the same feasible machine set and have the same proportional assignment. Therefore, let $C_{i,k}$ be the current processing time of all level $k$ jobs on machine $i$; note that $C_{i,k}$ will always be the same for all machines $i\in \cpM{k}$. We use a toy example to illustrate how our algorithm works as follows.

\textbf{Example.} Assume that there are $m=8$ machines submitting their speeds as $s_1=17$, $s_2=7$, $s_3=2$, and $s_4=\ldots=s_8=1$.
After the machines report their speeds, the algorithm first rounds down the speeds to the nearest power of $2$, resulting in $\bs_1=16$, $\bs_2=4$, $\bs_3=2$, $\bs_4=\ldots=\bs_8=1$. We will have $K=\lfloor\log m\rfloor+1=4$ groups of machines, where $\cM_1=\{1\},\cM_2=\emptyset,\cM_3=\{2\}$, and $\cM_4=\{3\}$ with group speed $r_1=16,r_2=8,r_3=4,r_4=2$. Note that the machines with rounded speeds $\bs_4,\ldots,\bs_8$ are ignored since $1 < \bs_1/m = 2$, and we will still have group $2$ with $r_2=8$ even if no machine belongs to it. 
Then, the first job $p_1=16$ arrives. We assign it to machine $1$ and set $\Lambda$ to $p_1/\bs_1=1$. Next, if a job $j$ with $p_j = 4$ arrives, the level of this job (i.e., $k(j)$) will be $3$ because $\Lambda \cdot r_3$ exactly equals $4$. We will assign $j$ proportional on $\cM_1 \cup \cM_2\cup\cM_3=\{1,2\}$. The fraction on machine $1$, denoted by $x_{1j}$, will be $4/5$, while $x_{2j}$ will be $1/5$, implying that $C_{1,3}=C_{2,3}=1/5$ after scheduling this job. 

The last interesting question concerns the timing of doubling. Generally, our algorithm doubles in any of the following two cases. The first scenario occurs when a super large job arrives, i.e., $p_j > \Lambda \cdot r_1$; the second arises when the first machine receives too many jobs from a specific level, i.e., $C_{1,k} > \Lambda$. 
Note that $C_{1,k}$ is equal to $C_{i,k}$ for every $i\in \cpM{k}$, so $C_{1,k} > \Lambda$ implies that every $C_{i,k} > \Lambda$. Both conditions suffice to indicate that $\opt > \Lambda$. 
A natural time to double is immediately after we observe one of these two cases. For instance, if a super large job arrives, we should double $\Lambda$ and then schedule it. However, this natural approach would violate both sides of monotonicity. To preserve the monotonicity, we implement doubling the guessed optimum $\Lambda$ with the following two crucial features, referred to as \emph{allocate-before-doubling} and \emph{double-without-the-last}.
\begin{itemize}
    \item Allocate-Before-Doubling (for Job-Side Monotonicity): We always schedule the job first, based on the old $\Lambda$, and then decide whether we should double $\Lambda$. Note that without this approach, if we schedule the job after doubling, the job-side monotonicity may be violated. Consider a job that can increase the size by an amount of $\epsilon$ to trigger a doubling procedure. If we allocate this job after doubling, its level may increase, which means we will introduce some additional slower machines to process it. Because of our proportional allocation, the unit processing time of $j$ will decrease, which breaks the job-side monotonicity. See 
    \cref{sec:abd} 
    for a complete example. 
    \item Double-Without-The-Last (for Machine-Side Monotonicity): The last level does not affect doubling; that is, even if $C_{1,K} > \Lambda$, we do not double. 
    The motivation of this operation is to prove an important property of our algorithm: the $\Lambda$-stability. This property states that when we increase a machine's speed from $s_i$ to $s_i' = 2 s_i$, the guessed value $\Lambda$ at every time will only decrease but not significantly, i.e., $\Lambda \geq \Lambda' \geq \Lambda/2$. However, the property will fail if we consider the last level for doubling. We will emphasize the reason again in the proof of the $\Lambda$-stability. 
\end{itemize}

We claim that these two properties only degrade the competitive ratio by a constant factor compared to the natural approach while preserving monotonicity. For the complete algorithm and all its details, please refer to \Cref{alg:makespan}. Before analyzing monotonicity and the competitive ratio, we claim a straightforward but useful property below. Let $\Lambda_j$ denote the guessed \opt ($\Lambda$) at the time of job $j$'s arrival.


\begin{algorithm}[tb]
\caption{Level-Based Proportional Allocation for Makespan}
\label{alg:makespan}
\KwIn{The reported speeds $\{s_i\}_{i\in [m]}$ and the reported job sizes $\{p_j\}_{j\in [n]}$ which shows up online.}
\KwOut{ An online fractional allocation $\vX = \{x_{ij}\}_{i\in [m],j\in [n]}$ }
Sort $\{s_i\}_{i\in [m]}$ in descending order \;
$\forall i\in [m]$, set $\bs_i \gets \max \{ 2^z \mid 2^z\leq s_i,~z\in \Z\}$ \tcp*{round down $s_i$ to the nearest power of $2$}

Ignore the machines with $\bs_i < \bs_1/m$ and partition the remaining machines into $K=\lfloor \log m \rfloor+1$ groups: $\forall 1\leq k \leq K$, let $\cM_k \gets \{ i\in [m] \mid \bs_i = r_k \}$ and $\cpM{k} \gets \bigcup_{k'\leq k} \cM_{k'}$, where $r_k:= \frac{\bs_1}{2^{k-1}}$ \;



\For{each arriving job $j$ }
{
\eIf{$j=1$}
{
$\forall i\in \cM_1$, set $x_{ij}\gets \frac{1}{|\cM_1|}$ \;

$\Lambda \gets p_1/r_1$   \tcp*{initialize a guessed optimal objective value}

$C_{i,k} \gets 0$ for all $k \in [K]$ and $i \in \cpM{k}$\tcp*{$K=\lfloor \log m \rfloor+1$}

}
{Define $\kj \gets \max \bigg(\{k \in [K]  \mid p_j \leq r_k \cdot \Lambda \} \cup \{1\}\bigg) $ \tcp*{decide job $j$'s level}

\For(\tcp*[f]{allocate-before-doubling}){each $i\in \cpM{\kj}$}
{
    $x_{ij} \leftarrow \frac{\bs_i}{\sum_{t\in \cpM{\kj}} \bs_{t}}$ \;
    $C_{i,\kj} \leftarrow C_{i,\kj} + x_{ij}\cdot \frac{p_j}{\bs_i}$ \; 
}
\If(\tcp*[f]{double-without-the-last}){ $k(j) \leq K-1$ } 
{
\If(\tcp*[f]{a super large job}){$p_j / r_1 > \Lambda $ }{Keep doubling $\Lambda$ until it is at least $p_j/r_1$ \; Reset $C_{i,k} \gets 0$ for all $k \in [K]$ and $i \in \cpM{k}$ \tcp*{start a new phase}}
\ElseIf(\tcp*[f]{a saturated level}){ $ C_{1,\kj} > \Lambda$ }{ Double $\Lambda$ \; Reset $C_{i,k} \gets 0$ for all $k \in [K]$ and $i \in \cpM{k}$ \tcp*{start a new phase}}

}
}



}
\Return{$\vX = \{x_{ij}\}_{i\in [m],j\in [n]}$.}

\end{algorithm}

\begin{observation}
\label{lem:makespan:lambda-form}
All $\Lambda_j$ can only take the form of $p_1$ multiplied by an integer power of 2, i.e. $\Lambda_j \in \{p_1\cdot 2^z \mid z\in \Z\}$. 
\end{observation}



\subsection{Monotonicity}



This subsection aims to show the monotonicity of the algorithm.

\begin{lemma} [Two-Sided Monotonicity]\label{lem:makespan:monontone}
The allocation returned by~\cref{alg:makespan} is both machine-side and job-side monotone.
\end{lemma}

\begin{proof}
    The proof of the job-side monotonicity is straightforward due to the ``allocate-before-doubling'' trick. When each job arrives, the algorithm first determines its level $\kj$, distributing it across the machines in $\cpM{\kj}$ proportionally, and then considers whether to double the value of $\Lambda$. Thus, increasing the reported size $p_j$ can only result in a smaller value of $\kj$ (which may not be the case if we first double $\Lambda$ according to the reported job size) and a faster average machine speed of $\cpM{\kj}$, which implies that the allocation is job-side monotone.

    For the machine side, we show the monotonicity by proving that for each job, its fraction scheduled on any machine is non-decreasing when the machine increases its speed. Consider an arbitrary machine $i^*$ (without loss of generality, we assume $i^*$ is the smallest index in its group for future notational convenience). Clearly, if the reported speed $s_{i^*}$ increases but is still rounded down to the same $\bs_{i^*}$, the algorithm's allocation remains unchanged. Therefore, the proof only needs to focus on the scenario that the machine increases its reported speed, causing an increment of the rounded speed from $\bs_{i^*}$ to $\bs_{i^*}'=2\bs_{i^*}$.
    Larger jumps in the rounded speed, i.e., from $\bs_{i^*}$ to 
    $2^k \bs_{i^*}$ for any $k \geq 2$, can be handled by applying the two-speed 
    case repeatedly: monotonicity from $\bs_{i^*}$ to $2\bs_{i^*}$ and from 
    $2\bs_{i^*}$ to $4\bs_{i^*}$ together imply monotonicity from $\bs_{i^*}$ to 
    $4\bs_{i^*}$, and the argument extends inductively to any $k$. Hence, it 
    suffices to establish the result for a single doubling of the rounded speed.
    
    
    To clarify the notations, we view the story before $i^*$ speeds up and after $i^*$ speeds up as two different worlds --- the \emph{original world} and the \emph{speed-up world}. 
    We use the notations without primes to denote the states in the original world, and the notations with primes to denote the states in the speed-up world. 
    For example, we let 
    $$
        \bs_i' = 
        \begin{cases}
            \bs_i, &  i \neq i^*\\
            2\bs_i, &  i = i^*~;
        \end{cases}
    $$ 
    let $C_{i,k}$ and $C_{i,k}'$ denote the processing time on machine $i$ contributed by jobs from level $k$ in the two worlds, respectively;  
    similarly, we define $\{r_k,r_k'\}$, $\{ \kj,k'(j)\}$, $ \{\cpM{\kj},\cpM{k'(j)}' \} $ and $\{x_{ij}, x_{ij}'\}$. Note that when the machine speed increases, not only the value of $\kj$ for each job $j$ may change, but the machine groups could also be affected. Our goal is to show that     \begin{equation}
        \label{eqn:each-job-increase}
        \forall j\in [n], \quad x_{i^*j} \leq x_{i^*j}'~.
    \end{equation}

     The inequality trivially holds for the first arriving job or the jobs where $x_{i^*j}=0$. For each remaining job $j$ with $x_{i^*j}>0$, we know that $i^* \in \cpM{\kj}$. Therefore, to prove that $x_{i^*j} \leq x_{i^*j}'$, we must first establish that after the speed increases, machine $i^*$ still remains part of the feasible machine set $\cpM{k'(j)}'$ of job $j$ (i.e., $i^*\in \cpM{k'(j)}'$); otherwise, $x_{i^*j}'$ would be $0$. Furthermore, according to the proportional allocation rule of the algorithm,  
     \[x_{i^*j} \leftarrow \frac{\bs_{i^*}}{\sum_{t\in \cpM{\kj}} \bs_{t}}
     \quad \text{ and } \quad 
     x_{i^*j}' \leftarrow \frac{\bs_{i^*}'}{\sum_{t\in \cpM{k'(j)}'} \bs_{t}'}~.\]
     Since $\bs_{i^*} < \bs_{i^*}'$, if we can demonstrate that $\cpM{k'(j)}' \subseteq \cpM{\kj}$, which means the feasible machine set for job $j$ does not gain any new machine due to the speed increase of machine $i^*$, then $x_{i^*j} \leq x_{i^*j}'$ can be proved. In summary, we need the following two properties for each job $j$ to prove Equation~\eqref{eqn:each-job-increase}:
    \begin{enumerate}
        \item[(\rom{1})]  If $i^*\in \cpM{\kj}$, then $i^*\in \cpM{k'(j)}'$~.
        \item[(\rom{2})] $\cpM{k'(j)}' \subseteq \cpM{\kj}$~.
    \end{enumerate}

    Whether a machine $i$ is feasible for job $j$ is determined by the job size $p_j$, machine speed $\bs_i$, and the guessed optimal value $\Lambda$ at the time of job $j$'s arrival (before allocating). The condition $i^*\in \cpM{\kj}$ implies that $p_j\leq \bs_{i^*}\cdot \Lambda_j$, where $\Lambda_j$ represents the value of $\Lambda$ at the time of job $j$'s arrival (before $i^*$ increases its speed), which is defined for $j\geq 2$. Thus, to prove property (\rom{1}), it is sufficient to show that $\Lambda_j'\geq \Lambda_j/2$, which leads to
    \[p_j\leq \bs_{i^*}\cdot \Lambda_j = 2\bs_{i^*}\cdot \frac{\Lambda_j}{2} \leq \bs_{i^*}'\cdot \Lambda_j'~,   \]
    and thereby completes the proof.

    For property (\rom{2}), since all other machines $i\neq i^*$ maintain their speeds, whether new machines become feasible for job $j$ only depends on whether $\Lambda_j$ increases. Specifically, if we can demonstrate that $\Lambda_j \geq \Lambda_j'$, then any machine $i\notin \cpM{\kj}$ satisfies that 
    \[ p_j > \bs_{i}\cdot \Lambda_j \geq \bs_{i}' \cdot \Lambda_j'~. \]
    This means that machine $i$ will not become feasible for job $j$ after machine $i^*$ increases the speed. 
    
    To conclude, we want $\Lambda_j $ to exhibit a certain level of monotonicity and stability. When the speed of machine $i^*$ is doubled, $\Lambda_j'$ is non-increasing but will never drop below $\Lambda_j/2$. We refer to this property as \emph{$\Lambda$-Stability} (\cref{lem:makespan:lambda_stable}). The proof of this property is the final piece in establishing the machine-side monotonicity. We will demonstrate that, owing to certain tricks employed in our algorithm (e.g., speed rounding and double-without-the-last), $\Lambda$ can be shown to be stable.
\end{proof}

Before formally showing the $\Lambda$-Stability, i.e., $ \displaystyle \Lambda_j \geq \Lambda_j' \geq \frac{1}{2}\Lambda_j$, we state some basic intuitions. In general, we first prove $\Lambda_j \geq \Lambda_j'$ and then use it as a black box to show $\Lambda_j' \geq \frac{1}{2}\Lambda_j$.  
To prove $\Lambda_j \geq \Lambda_j'$, there are two main steps. 
In the first step, we consider a \emph{good} job (time) $t$, following the definition below. 

\begin{definition}
A job (time) $t$ is called \emph{good} if 
\begin{itemize}
    \item The guessed $\opt$ is the same in both worlds, i.e., $\Lambda_t=\Lambda_t'$.
    \item The current completion time (every level) of the speed-up world  $C_{i,k}'=0$ for all $k \in [K]$ and $i \in \cpM{k}$ at the arrival time of job $t$ (before allocating $t$).
\end{itemize}
\end{definition}


We show that the next doubling time after a good job $t$ in the original world cannot be later than that in the speed-up world (\cref{lem:equaltime}), because machines in the speed-up world have a strictly faster speed. In the second step, we use this observation to prove the impossibility of $\Lambda_j < \Lambda_j'$. For example, the first job is straightforwardly a good job if $i^*$ is not in the first group. We can easily use it to prove $\Lambda_j > \Lambda_j'$ in the case when $j$ is the first time they differ. For later $j$, we will use the existence of another good job to prove. It is important to note that in the proof of $\Lambda$-stability, we critically rely on the double-without-the-last technique; without this approach, $\Lambda$-stability could be violated (see the 
\Cref{sec:dwtl} 
for an example).


\begin{lemma}
\label{lem:equaltime}
    Let $t$ be a good job (time), and $j$ be the first time $\Lambda_j' > \Lambda_t'$. We have $\Lambda_j > \Lambda_t$. 
\end{lemma}
\begin{proof}
    Assume for contradiction that $\Lambda_j = \Lambda_t$. According to the definitions of the good job and time $j$, we have
    \begin{equation}\label{eqn:contradict}
        \Lambda_j = \Lambda_{j-1} = \dots = \Lambda_t = \Lambda_{j-1}' = \dots = \Lambda_t'~.
    \end{equation}

    Since $j$ is the first time $\Lambda_j' > \Lambda_t'$, we know that $\Lambda'$ increases right after allocating job $j-1$.
    From the statement of~\cref{alg:makespan},  $\Lambda'$ may increase due to two reasons. If it is due to a super large job: $p_{j-1} / r_{1}' > \Lambda'$, then we have
    \[ \frac{p_{j-1}}{r_{1}} \geq \frac{p_{j-1}}{r_{1}'}   > \Lambda'_{j-1} = \Lambda_{j-1}  \]
    due to $r_{1}\leq r_{1}'$ and~\cref{eqn:contradict}. This implies that $\Lambda$ must be doubled before job $j$ arrives, which contradicts our assumption.

    The remaining case is that $\Lambda'$ increases due to a saturated level: $C_{1,k'(j-1)}' > \Lambda_{j-1}'$. Intuitively, it seems impossible that after allocating job $j-1$, $\Lambda'$ is doubled for this reason while $\Lambda$ has not been doubled, as machines in the speed-up world should be more powerful. By carefully analyzing the allocation of each job $h\in [t,j-1]$, we can finally draw a contradiction by showing that $\Lambda$ must be doubled right after allocating job $j-1$ due to \[C_{1,k(j-1)}\geq C_{1,k'(j-1)}' > \Lambda_{j-1}'=\Lambda_{j-1}~,\]
    which completes the proof.
    To this end, we claim the following two crucial properties.

    \begin{property}[Job-Level Consistency]\label{pro:makespan:job-consist}
         We call a job $h$ \emph{normal} if $\max\{k(h), k'(h)\} \leq K-1$ and $p_h \leq \min \{ r_1 \Lambda_h,~ r_1' \Lambda_h'\}$. For any two \emph{normal} jobs $h_1$ and $h_2$ with $ \Lambda_{h_1}=\Lambda_{h_2} $ and $\Lambda_{h_1}' = \Lambda_{h_2}'$, we have $k(h_1)=k(h_2)$ if and only if $k'(h_1)=k'(h_2)$. 
    \end{property}

    \begin{property}[Job-Makespan Monotonicity]\label{pro:makespan:job-mono}
        For each job $h$ with $\Lambda_h=\Lambda_h'$, when $k'(h)\leq K-1$ (not the last level),  $\displaystyle  \frac{x_{1h} }{\bs_1} \geq  \frac{x_{1h}'}{\bs_1'}$. 
    \end{property}

    \begin{figure}[tbp]
    \centering
    \includegraphics[width=0.6\linewidth]{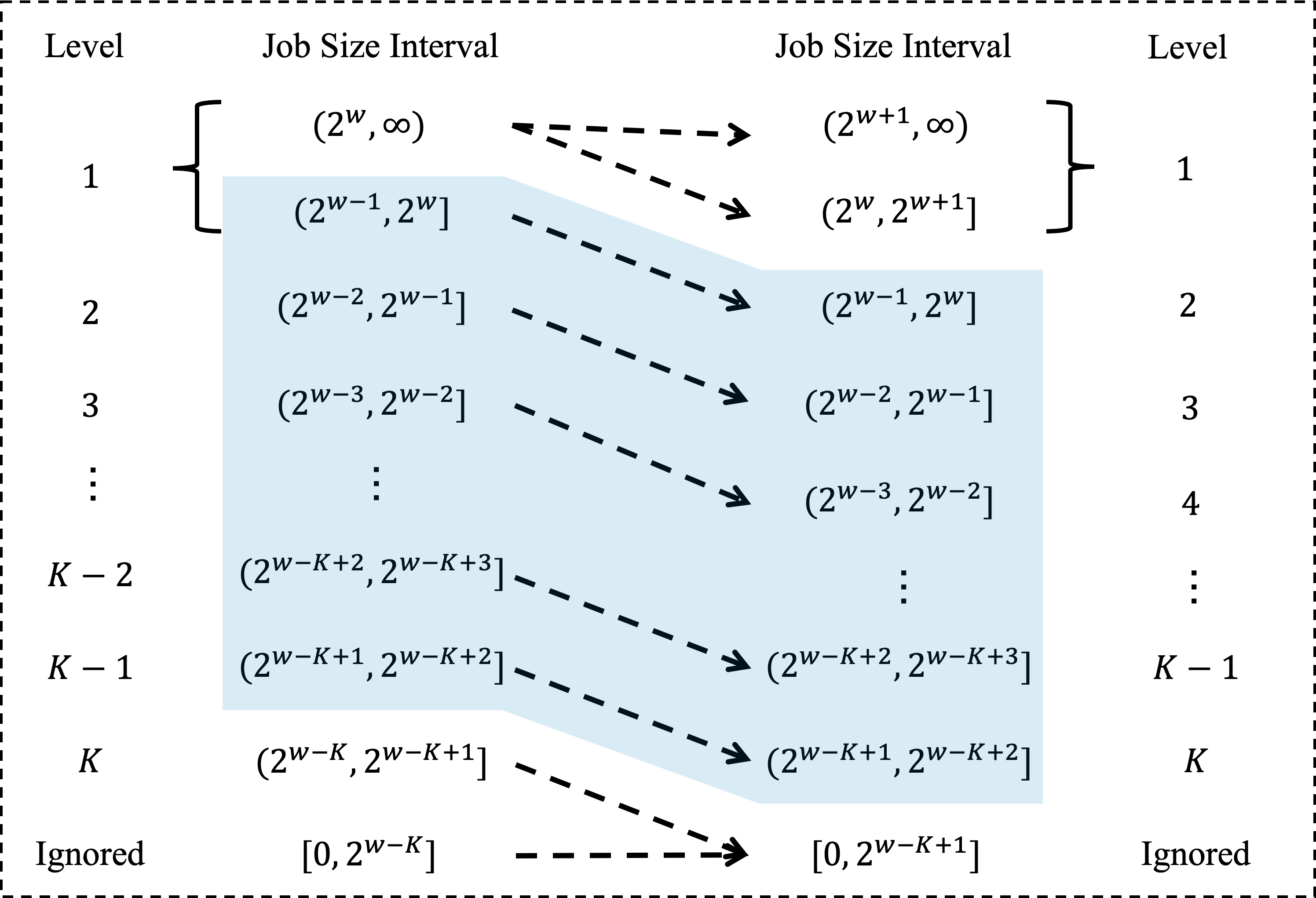}
    \caption{ An illustration of the job-level consistency property. The figure describes a scenario where the fastest machine increases the speed from $\bs_1=2^w$ (the left) to $\bs_1' =2\bs_1= 2^{w+1}$ (the right) and $\Lambda=\Lambda'=1$. As shown in the figure, the speed-up of the fastest machine leads to changes in machine groups, resulting in a shift in the levels assigned to jobs (even though the value of $\Lambda$ remains unchanged). The blue shaded area represents the normal jobs defined in~\cref{pro:makespan:job-consist}. We can observe that these jobs' levels have shifted down by one level overall; thus, jobs that originally belonged to the same level still remain at the same level in the speed-up world.  }
    \label{fig:jc3}
    \end{figure}


    The high-level idea of proving $C_{1,k(j-1)}\geq C_{1,k'(j-1)}'$ is to first show that $C_{1,k(j-1)}$ has a larger set of contributing jobs than $C_{1,k'(j-1)}'$ (by \cref{pro:makespan:job-consist}). 
    Then, for each of these contributing jobs, its contribution to $C_{1,k(j-1)}$ is no less than that to $ C_{1,k'(j-1)}'$ (by \cref{pro:makespan:job-mono}). We provide an illustration of~\cref{pro:makespan:job-consist} in~\cref{fig:jc3} and defer the proofs of these two properties to later. For now, we assume both properties are correct to complete the overall argument.




    
    Let us define $J=\{h\in [t,j-1]\mid k(h)=k(j-1)\}$ and $J'=\{h \in [t,j-1] \mid k'(h)=k'(j-1)\}$. As job $j-1$ triggers doubling and we use the double-without-the-last trick\footnote{Note that if the algorithm does not use the double-without-the-last trick, we cannot prove that job $j-1$ is normal, which results in the failure to apply~\cref{pro:makespan:job-consist}. }, we have $k'(j-1)\leq K-1$ and therefore, $k'(h)\leq K-1$ for any $ h\in J'$. Additionally, due to~\cref{eqn:contradict} and $r_i\leq r_i'$ for each $i$, we have $k(h) \leq k'(h)\leq K-1$ for any $h\in J'$.
    On the other hand, since the algorithm never performs doubling due to a super large job in both worlds, we can conclude that all jobs in $J'$ are normal. By~\cref{pro:makespan:job-consist}, for each job $h\in J'$, we have $k(h) = k(j-1)$ since $k'(h) = k'(j-1)$, and therefore, $J'\subseteq J$.
    Finally, combining with~\Cref{pro:makespan:job-mono}, right after allocating job ${j-1}$, we have 
    \begin{align*}
        C_{1,k(j-1)} - C_{1,k'(j-1)}' &\geq  \sum_{h\in J}x_{1h} \cdot \frac{p_h}{\bs_1}  - \sum_{h\in J'}x'_{1h} \cdot \frac{p_h}{\bs_1'} \tag{the second condition of good jobs} \\
        & \geq  \sum_{h\in J'}x_{1h} \cdot \frac{p_h}{\bs_1} - \sum_{h\in J'}x'_{1h} \cdot \frac{p_h}{\bs_1'} \tag{$J'\subseteq J$} \\
        & \geq  0~. \tag{\cref{pro:makespan:job-mono}}
    \end{align*}
\end{proof}

    
\begin{proof}[Proof of~\cref{pro:makespan:job-consist}]
    Before proving the lemma, we first provide some intuition. At the arrival of each job, we define a \emph{job size interval} for every level. For instance, a job $h$ is at level $3$ if and only if its size lies in $(\Lambda_{h}\cdot r_{4},\, \Lambda_{h}\cdot r_{3}]$; we call this range the job size interval of level $3$. Note that the parameters $\{\Lambda_h, r, k(h)\}$ in the original world and $\{\Lambda_h', r', k'(h)\}$ in the speed-up world may differ. For example, $k'(h)$ may become $4$ if $p_h$ falls into the level-$4$ interval in the speed-up world, namely $(\Lambda'_{h}\cdot r'_{5},\, \Lambda'_{h}\cdot r'_{4}]$. Nevertheless, we observe that for a normal job $h$, the two intervals containing $h$ are numerically identical: $\Lambda'_{h}\cdot r'_{5}=\Lambda_{h}\cdot r_{4}$ and $\Lambda'_{h}\cdot r'_{4}=\Lambda_{h}\cdot r_{3}$. In other words, except for the first and the last level, the job size intervals may shift in index but remain the same as numerical ranges in the two worlds; see~\cref{fig:jc3} for an illustration. Consequently, for two normal jobs $h_1$ and $h_2$, if they fall into the same job size interval in the original world, then they also fall into the same (numerically identical) job size interval in the speed-up world, even though their level indices may change.
    
    To formally prove it, since $\bs_1$ and $\bs_1'$ follow the rounding rule, we let $\bs_1=2^w,\bs_1'=2^{w'}$ where $w,w'\in \Z$. According to \cref{lem:makespan:lambda-form}, the guessed $\opt$ values can be represented by similar formulas: $\Lambda_{h_1}=\Lambda_{h_2}= 2^z\cdot p_1$ and $\Lambda_{h_1}'=\Lambda_{h_2}'=2^{z'}\cdot p_1$, where $z,z'\in \Z$. As $h_1$ is a normal\footnote{The job $h$ which is not normal may fall in a special interval such as $p_h\in (p_1\cdot 2^{z}\cdot r_{k(h)},+\infty]$ or $p_h\in [0,p_1\cdot 2^{z}\cdot r_{k(h)}]$. That is one main reason why we need to restrict $h_1$ and $h_2$ to be normal, and introduce the double-without-the-last technique.} job at level $k(h_1)$ and $k'(h_1)$ in the two worlds, 
    we have 
    \[p_{h_1}\in (\Lambda_{h_1}\cdot r_{k(h_1)+1},~\Lambda_{h_1}\cdot r_{k(h_1)}]=(p_1\cdot 2^{z+w-k(h_1)},~p_1\cdot 2^{z+w-k(h_1)+1}]~,\] 
    and 
    \[p_{h_1}\in (\Lambda_{h_1}'\cdot r_{k'(h_1)+1}',~\Lambda_{h_1}'\cdot r_{k'(h_1)}']=(p_1\cdot 2^{z'+w'-k'(h_1)},~p_1\cdot 2^{z'+w'-k'(h_1)+1}]~.\]
    Due to the integrality of $z+w-k(h_1)$ and $z'+w'-k'(h_1)$, the two intervals above intersect if and only if they are equal. Therefore, we get 
    \begin{align*}
        z+w-k(h_1)=z'+w'-k'(h_1)~.
    \end{align*}
    
    Similarly, for job $h_2$, we have 
    \begin{align*}
        z+w-k(h_2)=z'+w'-k'(h_2)~.
    \end{align*} 
    
    According to the above two equations, we can conclude that $k(h_1)=k(h_2)$ if and only if $k'(h_1)=k'(h_2)$.
\end{proof}

\begin{proof}[Proof of~\cref{pro:makespan:job-mono}]
 
    For each job $h$ with $\Lambda_h=\Lambda_h'$ and $k'(h) \leq K-1$\footnote{It is worth noting that this condition is crucial because, when $i=1$ and $\bs_1$ is increased to $2\bs_1$, the machines in $\cM_K$ would be ignored as their speeds are too slow to be considered.}, the set of feasible machines will only increase with machine acceleration, i.e., $\cpM{k(h)} \subseteq \cpM{k'(h)}'$.
    Then, according to our proportional allocation rule, we have
    \begin{align*}
          \frac{x_{1h}}{\bs_1} & = \frac{\bs_1}{\sum_{i\in \cpM{k(h)}}\bs_{i}} \cdot \frac{1}{\bs_1} \\
          & \geq  \frac{\bs'_1}{ \sum_{i\in \cpM{k'(h)}' }\bs_{i}'} \cdot \frac{1}{\bs'_1} \tag{$\cpM{k(h)} \subseteq \cpM{k'(h)}'$ and $\bs_i \leq \bs_i' $ \;$\forall i$}\\
          & = \frac{x'_{1h}}{\bs'_1}~. 
    \end{align*}
    \qedhere
\end{proof}

\begin{lemma}[$\Lambda$-Stability]\label{lem:makespan:lambda_stable}
    For any job $j\geq 2$, we have $\displaystyle \Lambda_j \geq \Lambda_j' \geq \frac{1}{2}\Lambda_j$.
\end{lemma}

\begin{proof}
    We start by proving $\Lambda_j \geq \Lambda_j'$. This inequality clearly holds for job $2$ as $\Lambda_2 = p_1/r_1 \geq p_1/r_1' =\Lambda_2'  $. Assume for contradiction that there exists a job $j$ which is the first job violating the inequality. This means that \begin{equation}\label{eqn:beforefirstjob}
        \forall t<j, \; \Lambda_t \geq \Lambda_t'~,
    \end{equation} and $\Lambda_j < \Lambda_j'$. Notice that due to~\cref{lem:makespan:lambda-form}, both $\Lambda_j'$ and $\Lambda_j$ are integer powers of $2$ multiplied by $p_1$. Therefore, $\Lambda_j < \Lambda_j'$ indicates that
    \begin{equation}\label{eqn:firstjob}
        \Lambda_j \leq \frac{1}{2}\Lambda_j'~.
    \end{equation}

    Since job $j$ is the first one, $\Lambda'$ must be increased right after allocating job $(j-1)$.  
    $\Lambda'$ may increase due to two reasons.
    We first show that this increment cannot occur due to a super large job: $p_{j-1}/r_1' > \Lambda_{j-1}'$.

    
    
    Supposing that $p_{j-1}/r_1' > \Lambda_{j-1}'$, as the algorithm keeps doubling until $\Lambda'$ is no less than $p_{j-1}/r_1'$, $\Lambda'_j$ must be the smallest value in $\{p_1\cdot 2^z \mid z\in \Z\}$ (\cref{lem:makespan:lambda-form}) that is at least $p_{j-1}/r_1'$. On the other hand, in the scenario before $i^*$ increases the speed, the algorithm always has $\Lambda_j \geq p_{j-1}/r_1$ due to our doubling criteria. Then, since $r_1 \leq r_1'$,
    \[ \Lambda_j \geq \frac{p_{j-1}}{r_1} \geq \frac{p_{j-1}}{r_1'}~, \]
    which implies that $\Lambda_j$ is also a value in $\{p_1\cdot 2^z \geq p_{j-1}/r_1' \mid z\in \Z \}$, leading to the conclusion that $\Lambda_j\geq \Lambda_j'$ and thereby resulting in a contradiction.
    
    Thus, if~\cref{eqn:firstjob} holds, the only reason for increasing $\Lambda'$ right after allocating job $j-1$ must be a saturated level: $C_{1,k'(j-1)}' > \Lambda_{j-1}'$. Since the algorithm only doubles $\Lambda'$ once in this case, we have $\Lambda_{j-1}' = \frac{1}{2} \Lambda_j'$.  

    
    Find the first job $t$ with 
    \begin{equation}\label{eqn:1-2}
        \Lambda_t' = \Lambda_{j-1}'=\frac{1}{2}\Lambda_j'~.
    \end{equation}
     Clearly, $\Lambda'$ increases right after allocating job $t-1$, and therefore, each $C_{i,k}'$ is reset to $0$ at the time of job $t$'s arrival.
    Then, we have
    \begin{equation}\label{eqn:equal}
       \Lambda_j \geq \Lambda_t \geq \Lambda_t' = \frac{1}{2} \Lambda_j'  \geq  \Lambda_j~,
    \end{equation}
    where the first inequality uses the fact that $t<j$, the second inequality uses~\cref{eqn:beforefirstjob} and $t<j$, the third equation uses~\cref{eqn:1-2}, and the last inequality is due to~\cref{eqn:firstjob}.
    Because the two sides of~\cref{eqn:equal} are the same, all the inequalities in-between must hold with equality. Combining~\cref{eqn:1-2} and~\cref{eqn:equal}, we have
    \begin{equation}\label{eq:10-1}
        \Lambda_t = \Lambda_t' = \Lambda_{j-1}' = \Lambda_{j-1} =\Lambda_j = \frac{1}{2} \Lambda_j'~.
    \end{equation} 

    Recalling the definition of good jobs, we see that $t$ is a good job because of \cref{eq:10-1} and each $C_{i,k}'$ being reset to $0$ at the time of job $t$'s arrival. Then, by~\cref{lem:equaltime}, we have $\Lambda_j > \Lambda_t$, drawing a contradiction to~\cref{eq:10-1} and completing the proof of $\Lambda_j \geq \Lambda_j'$.

    The inequality $\Lambda_j \geq \Lambda_j'$  can easily yield the corollary that for any two machine speed profiles $\{s_i^{(1)}\}_{i\in [m]}$ and $\{s_i^{(2)}\}_{i\in [m]}$, if one dominates the other (i.e., $s_i^{(1)}\geq s_i^{(2)}~\forall i\in [m]$ ), then $\Lambda_j^{(2)}\geq \Lambda_j^{(1)}$ for any job $j$. Building on this corollary, we can readily prove the other side of the lemma: $\Lambda_j' \geq \Lambda_j/2$. 

    Consider a world where all machine speeds are doubled.
    As our proportional allocation rule and the doubling criteria are scaling-free, the allocation for each job $j$ and the doubling time points remain unchanged; the only difference is that the corresponding $\Lambda_j$ is halved to $\Lambda_j/2$. Since this new speed profile dominates the profile where only machine $i^*$ has its speed doubled, we conclude that $\Lambda_j'\geq \Lambda_j/2$.
\end{proof}

\subsection{Competitive Analysis}

\cref{alg:makespan} returns an online fractional assignment $\vX$ that is two-sided implementable. Our final mechanism executes this algorithm and performs an online rounding to obtain an integral solution: for each job $j$, we independently assign it to machine $i$ with probability $x_{ij}$. As mentioned above, this independent rounding does not have any impact on the desired monotonicity in expectation of our mechanism. 

This subsection analyzes the competitive ratio of the random mechanism. Use $\vY=\{y_{ij}\in \{0,1\}\}$ to denote the random allocation by our final mechanism. Let $\obj(\vX)$, $\obj(\vY)$ and $\opt$ be the objectives corresponding to $\vX$, $\vY$, and the optimal offline solution, respectively.
We show the following:

\begin{lemma} [Competitive Ratio] \label{thm:makespan_ratio}
    With probability at least $1-1/m^2$, we have $\obj(\vY) \leq O(\log m) \cdot \opt$. Further, the expected competitive ratio of $\vY$ is also guaranteed to be $O(\log m)$. 
\end{lemma}




To streamline the proof of~\cref{thm:makespan_ratio}, we first show that an $\alpha$-competitive and speed-size-feasible fractional solution can, with high probability, yield a competitive ratio of $O(\log m)$, and then analyze the speed-size-feasibility and the competitive ratio of~\cref{alg:makespan}.

\begin{lemma}[Approximately Roundability]\label{lem:rounding}
   For any (online) fractional allocation $\vX$ that is speed-size-feasible, if $\obj(\vX) \leq \alpha \cdot \opt$, then the natural independent rounding yields a competitive ratio of $O\left(\max \left\{\log m, \alpha \right\}\right)$ with probability at least $1-1/m^2$. Furthermore, the rounding is guaranteed to be $O(m^2)$-competitive in the worst case.
\end{lemma}

\begin{proof}
    The basic proof idea is to first apply concentration on each machine and then employ the union bound across all machines. Use $\vY=\{y_{ij}\in \{0,1\}\}$ to denote the random allocation obtained through independent rounding. For each machine $i\in [m]$, let $\cJ_i$ be the subset of jobs that are assigned fraction $x_{ij}>0$ to machine $i$. For each job $j\in \cJ_i$, define a random variable $z_{ij} := y_{ij}\cdot p_j / s_i$.     
    As $\vX$ is speed-size-feasible, we further normalize each $z_{ij}$ to $z_{ij}':=\frac{z_{ij}}{c \cdot \opt}$ such that it belongs to the interval $[0,1]$.

    Since $\vX$ is $\alpha$-competitive, we have
    \[\E\left[\sum_{j\in \cJ_i}z'_{ij}\right]=\frac{\E\left[\sum_{j\in \cJ_i}z_{ij}\right]}{c \cdot \opt}=\frac{\sum_{j\in \cJ_i} x_{ij}\cdot \frac{p_j}{s_i}}{c \cdot \opt} \leq \frac{\alpha}{c} .  \]
    For simplicity, define $\mu_i:=\E\left[\sum_{j\in \cJ_i}z'_{ij}\right] $. Given any $\delta\geq 2$, applying the Chernoff bound for machine $i$,
    \begin{align*}
        \Pr\left[ \sum_{j\in \cJ_i}z'_{ij} - \mu_i \geq \delta \cdot \frac{\alpha}{c} \right] 
        & \leq \exp{\left(-\frac{\delta \alpha}{2c}\right)} = \frac{1}{m^3} \cdot \exp{\left(-\frac{\delta \alpha}{2c} + 3\log m\right)} . 
    \end{align*}

    By setting $\delta =\max\left\{\frac{6c\log m}{\alpha},2\right\}$ and applying the union bound, we have with probability at least $1-1/m^2$, any machine $i$ satisfies
    \[ \sum_{j\in \cJ_i} y_{ij} \cdot \frac{p_j}{s_i} \leq \sum_{j\in \cJ_i} x_{ij}\cdot \frac{p_j}{s_i}  + \delta \cdot \alpha \cdot  \opt \leq O(\max\{\log m, \alpha\}) \cdot \opt. \]

    Finally, as $\vX$ is speed-size-feasible, even if in the worst case where all jobs are assigned to the machine with speed $s_1/m$, the objective is at most $m^2\cdot \opt$. Thus, $\vY$ is always $O(m^2)$-competitive.
\end{proof}

Now we show that the solution returned by~\cref{alg:makespan} is speed-size-feasible and bounded-competitive by~\cref{lem:makespan:roundable} and~\cref{lem:makespan:fractionalcompetitive}, respectively.

\begin{lemma}\label{lem:makespan:roundable}
    The fractional solution $\vX$ returned by~\cref{alg:makespan} is speed-size-feasible.
\end{lemma}
\begin{proof}
According to \Cref{def:speed_constraints}, a fractional allocation $\vX$ is speed-size-feasible if it satisfies the job-size constraint $p_j\le c\cdot s_i\cdot \opt, \forall j\in [n],i\in [m]$ where $c$ is a constant and the machine-speed constraint $s_i\ge \frac{s_1}{c\cdot m}, \forall i\in [m]$ where $s_1$ is highest machine speed.
    From the description, the algorithm initially scales down the machine speeds and ignores all machines with $\bs_i < \frac{\bs_1}{m}$.
    Since $s_i\leq 2\bs_i$ and $\bs_1 \leq s_1$, the machine-speed constraint in~\cref{def:speed_constraints} has been satisfied. 
    
    To show that $\vX$ satisfies the job-size constraint, we first note that the initial operation can be regarded as the algorithm creating a new machine set $\bcM$ and allocating online jobs to this new set of machines. It is easy to observe that the optimal solution $\bopt$ on this new machine set will only have a slight increase: $\bopt \leq 4\cdot \opt$.

    Each machine in $\cM_1$ must satisfy the job-size constraint as any job $j\in \cJ$ has $p_j / \bs_1 \leq \bopt$. For each other machine $i\notin \cM_1$, according to our job-level criteria, the assigned job's size is bounded by $\bs_i \cdot \Lambda$. We now show that the final guessed objective $\Lambda_n$ is at most $2\cdot \bopt$.

    We assume for contradiction that $\Lambda_n > 2\cdot \bopt$. Find the last job $j$ with $\Lambda_j \leq 2\cdot \bopt$. Clearly, the increase of $\Lambda$ right after allocating job $j$ must be due to $C_{1,\kj} > \Lambda_j$; otherwise, $\Lambda_{j+1}$ cannot be greater than $2\cdot \bopt$.
    Thus, we have $\Lambda_j = \Lambda_{j+1}/2 \geq \bopt$. 
    
    Use $\cJ (\kj) := \{h \in [n] \mid k(h)=\kj \text{ and } \Lambda_h = \Lambda_j \}$ to denote the set of jobs that have the same level and corresponding $\Lambda$ as job $j$. Our algorithm distributes these jobs proportionally across $\cpM{\kj}$. According to the definition of a job's level, each job $j'\in \cJ (\kj)$ satisfies $p_{j'} / \bs_i > \Lambda_j \geq\bopt$ for any machine $i\notin \cpM{\kj}$, implying that the optimal solution must allocate job $j'$ to a machine in $\cpM{\kj}$. Let $\pi(i)$ be the set of jobs in $\cJ (\kj)$ that are assigned to machine $i$ in the optimal solution. Clearly, for each machine $i\in \cpM{\kj}$, 
    \[    \bs_i \cdot \bopt \geq  \sum_{j\in \pi (i)} p_j. \]
    Summing over all these machines, we have
    \begin{align*}
       \sum_{i\in \cpM{\kj}}\bs_i \cdot \bopt & \geq \sum_{i\in \cpM{\kj}}\sum_{j\in \pi (i)} p_j   
       \bopt \geq  \frac{\sum_{j\in \cJ (\kj)} p_j}{ \sum_{i\in \cpM{\kj}}\bs_i }   \\
        & = \sum_{j\in \cJ (\kj)} \frac{p_j}{\bs_1} \cdot \frac{\bs_1}{\sum_{i\in \cpM{\kj}}\bs_i } = C_{1,\kj},
    \end{align*}
    which contradicts the fact that $\Lambda$ increases due to $C_{1,\kj} > \Lambda_j$ and completes the proof.
\end{proof}

\begin{lemma}\label{lem:makespan:fractionalcompetitive}
    The fractional solution $\vX$ returned by~\cref{alg:makespan} has $\obj(\vX) \leq O(\log m) \cdot \opt$.
    
\end{lemma}

\begin{proof}

    In this proof, we also assume that the algorithm assigns jobs to the scaled machine set $\bcM$ with $\bopt \leq 4\cdot \opt$.
    For each arriving job $j$, the algorithm sets $x_{ij}=\frac{\bs_i}{\sum_{i'\in \cpM{\kj}} \bs_{i'} } $ for each machine $i\in \cpM{\kj}$, implying that this job shares the same processing time $ \frac{p_j}{\sum_{i'\in \cpM{\kj}} \bs_{i'} } $ on each machine~$i$. As the first machine belongs to any $\cpM{\kj}$, we have 
    \[\obj(\vX) \leq \sum_{j\in \cJ} x_{1j}\cdot \frac{p_j}{\bs_1} = \sum_t \sum_{j\in \cJ^{(t)}} x_{1j}\cdot \frac{p_j}{\bs_1}, \]
    where $ \cJ^{(t)} $ is the job subset in the $t$-th phase of the algorithm. Here, we use the $t$-th phase to denote the time period between the $(t-1)$-th doubling and the $t$-th doubling.
    
    The total contribution of the level-$K$ jobs across all phases can be easily bounded by $\bopt$, since we assign each of them to all machines proportionally to machine speeds, achieving a lower bound on the optimal solution.

    For other jobs, we analyze their contribution for each phase.
    Use $\Lambda^{(t)}$ to denote the guessed $\Lambda$ value in phase $t$. According to Line 14 in~\cref{alg:makespan}, before assigning the last job in phase $t$, each job level's total processing time $C_{1,k}$ is at most $\Lambda^{(t)}$ for $1 \leq k \leq K -1$; while the processing time of the last job $j$ in phase $t$ is bounded by the guessed optimal objective of the next phase: $p_j /r_1 \leq \Lambda^{(t+1)}$. Since there are $K-1 \leq \lfloor \log m \rfloor$ job levels,  their total contribution in phase $t$ is at most 
    \[ \lfloor \log m \rfloor \cdot \Lambda^{(t)}+\Lambda^{(t+1)},   
    \]

    As $\Lambda^{(t)}$ increases exponentially and the final $\Lambda$ is at most $2\cdot \bopt$ (see the proof of~\cref{lem:makespan:roundable}), we have $\obj(\vX) \leq  O(\log m) \cdot \opt$.
\end{proof}

\begin{proof}[Proof of~\cref{thm:makespan_ratio}]
    The theorem can be proved directly by the above lemmas.
    \cref{lem:makespan:fractionalcompetitive} and~\cref{lem:makespan:roundable} prove that the returned fractional allocation $\vX$ is $O(\log m)$-competitive and speed-size-feasible.  
     \cref{lem:rounding} demonstrates that for such a fractional allocation, independent rounding returns a $O(\log m)$-competitive solution with probability at least $1-1/m^2$ and always remains $O(m^2)$-competitive, implying an expected competitive ratio of $O(\log m)$. 
\end{proof}


\subsection{Proof of \Cref{thm:makespan}}
Finally, we conclude the proof of \Cref{thm:makespan}. We use an independent rounding method to round the fractional algorithm in \Cref{alg:makespan}. Because of the independent rounding, the expected load of a machine is the same as the load of a machine returned by the fractional algorithm. Therefore, the two-sided monotonicity of \Cref{alg:makespan} is the same as the two-sided monotonicity of the randomized algorithm after rounding. Therefore, the randomized algorithm is two-sided monotone, and thus two-sided implementable by \Cref{lem:machine-side-implementable} and \Cref{lem:job-side-implementable}. For competitive ratios, in \Cref{thm:makespan_ratio}, we have already shown the competitive ratio in expectation and with high probability. Finally, we briefly discuss why our final mechanism that implements our randomized algorithm runs in polynomial time. 

Our randomized algorithm clearly runs in polynomial time. Every time a job $j$ arrives, we compute $x_{ij},i\in \cpM{k(j)}$ for at most $m$ machines, each $x_{ij}$ can be calculated in polynomial time. Then we check whether some $C_{1,k},k\in [K]$ exceed $\Lambda$, which can also be done in $\log m$ times of comparisons. 

Next, after we use a mechanism to implement our randomized algorithm, we need to claim that the job and machine payments can be calculated in polynomial time. According to \cref{lem:job_truth}, the expected price for a job $j$ who reports its size $p$ is 
\begin{align*}
    Q_j(p) = Q_j(0) - \sum_{i} \left(C_i^{(j)}\cdot \left(x_{ij}(p)-x_{ij} (0)\right) + \frac{1}{s_i} \cdot \left(p \cdot x_{ij}(p) - \int_{0}^{p}x_{ij}(t) \mathrm{d}t \right)\right).
\end{align*}
After a job report $p$, because the calculation of all $x_{ij}$ can be done in polynomial time, the price can be computed in polynomial time if the integral can be computed in polynomial time.
Since we divide the jobs and their corresponding machine sets into $\log m$ number of levels, there are at most $\log m$ kinds of different $\sum_{i}x_{ij}(p)$. Therefore, the integral can be computed in polynomial time. We remark here that even if a selfish job wants to decide whether to report the real size, it can be done in polynomial time. It is because we only have $\log m$ different allocation for it, and the function $Q_j(p)$ also only has $\log m$ different values. 

Finally, we claim that the calculation of payments for machines is in polynomial time. According to \cite{DBLP:conf/focs/ArcherT01}, the expected payment of a machine $i$ who report its speed $s_i=1/b$ is 
\begin{align*}
    P(b)=b{L(b)}+\int_{b}^{+\infty}L(t)\mathrm{d} t
\end{align*}
where $L(b)$ is the load machine $i$ gets if it reports its speed $s_i=1/b$.
If the integral can be computed in polynomial time, the payment can be computed in polynomial time. Consider the fastest machine $i'$ despite machine $i$. Since \cref{alg:makespan} satisfies \cref{def:speed_constraints} and all the speeds are rounded into $\log m$ number of groups, there are at most $2\log m$ kinds of different $L(b)$: When $\bs_i>2^{\log m}\cdot\bs_{i'}$, all the jobs will be allocated on machine $i$; When $\bs_i<2^{-\log m}\cdot\bs_{i'}$, there's no job scheduled on machine $i$. Therefore, we can execute the polynomial time algorithm at most $2\log m$ times to compute each $L(b)$. In conclusion, the integral can be computed in polynomial time. Furthermore, if $Q_j(p) \geq 0$ is required to strictly qualify as a valid payment, we need to set a proper $Q_j(0)$ unrelated to $p$ as proved in \cref{lem:job_truth}.


\section{Generalization to $\ell_q$ Norm Objective}\label{sec:norm}

This section extends our techniques to the $\ell_q$ norm objective and shows the following:

\begin{theorem}\label{thm:lq}
    For online two-sided selfish scheduling with $\ell_q$ norm minimization, there exists a randomized polynomial time mechanism that is two-sided truthful and yields a competitive ratio of $O(m^{\frac{1}{q}(1-\frac{1}{q})} \cdot (\log m)^{1+\frac{1}{q^2}} )$ both in expectation and with high probability.
\end{theorem}

The mechanism shares the same framework as \Cref{alg:makespan}, with details provided in 
\Cref{alg:lq} (in \Cref{sec:lq-analysis}).
We highlight two differences here, one is the proportional assignment. $x_{ij}$ is changed from $\frac{\bs_i}{\sum_{i'\in \cpM{\kj}} \bs_{i'}}$ to $\frac{\bs_i^{\gamma}}{\sum_{i'\in \cpM{\kj}} \bs_{i'}^{\gamma}}$, where $\gamma = \frac{q}{q-1}$. When $q=\infty$, it reduces to the assignment in the makespan version. When $q=1$, $\gamma \rightarrow \infty$, we can view the assignment as scheduling all the jobs to the fastest machine. We do not rigorously include this case ($q=1$, $\gamma \rightarrow \infty$) in every proof because of the "not well-defined" issue. However, it is straightforward to see the algorithm is two-sided monotone, and gets an optimal solution when $q=1$. The other difference is the doubling criteria, changed from $C_{1,k(j)} > \Lambda$ to a generalized version of $|| \vC_{\kj} ||_q > \Lambda$. Here $\vC_{k} :=  \langle C_{i,k} \rangle_{i\in \cpM{k}}$ means the vector of total processing time of level $k$ jobs on machines in $\cpM{k}$.

Similar to the previous section, we show that the fractional algorithm is speed-size-feasible and two-sided monotone. For the two-sided monotonicity, we strictly follow the same framework as the analysis in the makespan version. We observe that almost all the properties can be safely generalized into the $\ell_q$ norm setting. Therefore, we present counterparts for all lemmas and observations in \Cref{sec:makespan} to prove the two-sided monotonicity. For competitive analysis, we present a new approximately rounding lemma, which introduces more loss compared to the makespan version in 
\Cref{lem:lq:rounding}
, thus proving the new competitive ratio. We defer all the analyses to
~\Cref{sec:lq-analysis}.

\section{Conclusion}\label{sec:con}
In conclusion, this paper presents a novel online truthful mechanism on the machine side and further extends it to the two-sided truthful setting, addressing general $\ell_q$ norm minimization. We achieve an $O(\log m)$-competitive ratio for makespan and a ratio of $\tilde{O}\bigl(m^{\frac{1}{q}(1-\frac{1}{q})}\bigr)$ for the $\ell_q$ norm. 
Our algorithm follows a general framework that combines a speed-size-feasible fractional algorithm with independent rounding. The key requirement is to ensure two-sided monotonicity for the fractional component. However, due to the use of independent rounding, the best attainable ratio is $O(\log m / \log \log m)$, as inherited from the classical balls-into-bins setting. Thus, achieving a constant competitive ratio would require the development of fundamentally new techniques.

Finally, we pose three open questions that we find particularly intriguing:
\begin{itemize}[left=2em]
    \item Can we design an $O(1)$-competitive machine-side truthful algorithm? Furthermore, is it possible to extend this to obtain two-sided truthfulness?
    \item Can we establish lower bounds on the competitive ratio when both online constraints and machine-side truthfulness are required? Existing lower bounds are typically proven either for online algorithms or for truthful mechanisms individually. To obtain stronger lower bounds—beyond small constant factors—it may be necessary to develop new techniques that simultaneously capture these two aspects. Even improving the constant in known bounds would be valuable if it explicitly combines the requirements of online and truthfulness. Furthermore, if we also impose two-sided truthfulness, could this lead to even stronger lower bounds?
    \item Can we achieve an $O(\log m)$-competitive ratio using a two-sided truthful mechanism (or machine-side truthful mechanism) for the $\ell_q$ norm objective? The ratio of the $\ell_q$-norm is not $O(\log m)$ in our algorithm. This is mainly because, in scenarios where a large number of slow machines are feasible to handle a large job, we assign the job proportionally among all of them, whereas \opt might only activate one machine. There are several intuitive approaches to address this issue, but none preserve machine-side monotonicity. 
\end{itemize}



\newpage
\bibliographystyle{plainnat}
\bibliography{ref}

@article{DBLP:journals/tcs/AulettaPPP09,
  author       = {Vincenzo Auletta and
                  Roberto De Prisco and
                  Paolo Penna and
                  Giuseppe Persiano},
  title        = {On designing truthful mechanisms for online scheduling},
  journal      = {Theor. Comput. Sci.},
  volume       = {410},
  number       = {36},
  pages        = {3348--3356},
  year         = {2009}
}

@book{archer2004mechanisms,
  title={Mechanisms for discrete optimization with rational agents},
  author={Archer, Aaron Francis},
  year={2004},
  publisher={Cornell University}
}

@inproceedings{DBLP:conf/focs/ArcherT01,
  author       = {Aaron Archer and
                  {\'{E}}va Tardos},
  title        = {Truthful Mechanisms for One-Parameter Agents},
  booktitle    = {42nd Annual Symposium on Foundations of Computer Science, {FOCS} 2001,
                  14-17 October 2001, Las Vegas, Nevada, {USA}},
  pages        = {482--491},
  publisher    = {{IEEE} Computer Society},
  year         = {2001},
  url          = {https://doi.org/10.1109/SFCS.2001.959924},
  doi          = {10.1109/SFCS.2001.959924},
  bibsource    = {dblp computer science bibliography, https://dblp.org}
}

@article{DBLP:journals/siamcomp/DhangwatnotaiDDR11,
  author       = {Peerapong Dhangwatnotai and
                  Shahar Dobzinski and
                  Shaddin Dughmi and
                  Tim Roughgarden},
  title        = {Truthful Approximation Schemes for Single-Parameter Agents},
  journal      = {{SIAM} J. Comput.},
  volume       = {40},
  number       = {3},
  pages        = {915--933},
  year         = {2011},
  url          = {https://doi.org/10.1137/080744992},
  doi          = {10.1137/080744992},
  bibsource    = {dblp computer science bibliography, https://dblp.org}
}

@inproceedings{DBLP:conf/stacs/AndelmanAS05,
  author       = {Nir Andelman and
                  Yossi Azar and
                  Motti Sorani},
  editor       = {Volker Diekert and
                  Bruno Durand},
  title        = {Truthful Approximation Mechanisms for Scheduling Selfish Related Machines},
  booktitle    = {{STACS} 2005, 22nd Annual Symposium on Theoretical Aspects of Computer
                  Science, Stuttgart, Germany, February 24-26, 2005, Proceedings},
  series       = {Lecture Notes in Computer Science},
  volume       = {3404},
  pages        = {69--82},
  publisher    = {Springer},
  year         = {2005},
  url          = {https://doi.org/10.1007/978-3-540-31856-9\_6},
  doi          = {10.1007/978-3-540-31856-9\_6},
  bibsource    = {dblp computer science bibliography, https://dblp.org}
}

@inproceedings{DBLP:conf/esa/Kovacs05,
  author       = {Annam{\'{a}}ria Kov{\'{a}}cs},
  editor       = {Gerth St{\o}lting Brodal and
                  Stefano Leonardi},
  title        = {Fast Monotone 3-Approximation Algorithm for Scheduling Related Machines},
  booktitle    = {Algorithms - {ESA} 2005, 13th Annual European Symposium, Palma de
                  Mallorca, Spain, October 3-6, 2005, Proceedings},
  series       = {Lecture Notes in Computer Science},
  volume       = {3669},
  pages        = {616--627},
  publisher    = {Springer},
  year         = {2005},
  url          = {https://doi.org/10.1007/11561071\_55},
  doi          = {10.1007/11561071\_55},
  bibsource    = {dblp computer science bibliography, https://dblp.org}
}

@article{DBLP:journals/jda/Kovacs09,
  author       = {Annam{\'{a}}ria Kov{\'{a}}cs},
  title        = {Tighter approximation bounds for {LPT} scheduling in two special cases},
  journal      = {J. Discrete Algorithms},
  volume       = {7},
  number       = {3},
  pages        = {327--340},
  year         = {2009},
  url          = {https://doi.org/10.1016/j.jda.2008.11.004},
  doi          = {10.1016/J.JDA.2008.11.004},
  bibsource    = {dblp computer science bibliography, https://dblp.org}
}

@article{DBLP:journals/siamcomp/0001K13,
  author       = {George Christodoulou and
                  Annam{\'{a}}ria Kov{\'{a}}cs},
  title        = {A Deterministic Truthful {PTAS} for Scheduling Related Machines},
  journal      = {{SIAM} J. Comput.},
  volume       = {42},
  number       = {4},
  pages        = {1572--1595},
  year         = {2013},
  url          = {https://doi.org/10.1137/120866038},
  doi          = {10.1137/120866038},
  bibsource    = {dblp computer science bibliography, https://dblp.org}
}

@article{DBLP:journals/mor/EpsteinLS16,
  author       = {Leah Epstein and
                  Asaf Levin and
                  Rob van Stee},
  title        = {A Unified Approach to Truthful Scheduling on Related Machines},
  journal      = {Math. Oper. Res.},
  volume       = {41},
  number       = {1},
  pages        = {332--351},
  year         = {2016},
  url          = {https://doi.org/10.1287/moor.2015.0730},
  doi          = {10.1287/MOOR.2015.0730},
  bibsource    = {dblp computer science bibliography, https://dblp.org}
}

@inproceedings{DBLP:conf/stoc/NisanR99,
  author       = {Noam Nisan and
                  Amir Ronen},
  editor       = {Jeffrey Scott Vitter and
                  Lawrence L. Larmore and
                  Frank Thomson Leighton},
  title        = {Algorithmic Mechanism Design (Extended Abstract)},
  booktitle    = {Proceedings of the Thirty-First Annual {ACM} Symposium on Theory of
                  Computing, May 1-4, 1999, Atlanta, Georgia, {USA}},
  pages        = {129--140},
  publisher    = {{ACM}},
  year         = {1999},
  url          = {https://doi.org/10.1145/301250.301287},
  doi          = {10.1145/301250.301287},
  bibsource    = {dblp computer science bibliography, https://dblp.org}
}

@article{DBLP:journals/algorithmica/ChristodoulouKV09,
  author       = {George Christodoulou and
                  Elias Koutsoupias and
                  Angelina Vidali},
  title        = {A Lower Bound for Scheduling Mechanisms},
  journal      = {Algorithmica},
  volume       = {55},
  number       = {4},
  pages        = {729--740},
  year         = {2009},
  url          = {https://doi.org/10.1007/s00453-008-9165-3},
  doi          = {10.1007/S00453-008-9165-3},
  bibsource    = {dblp computer science bibliography, https://dblp.org}
}

@article{DBLP:journals/algorithmica/KoutsoupiasV13,
  author       = {Elias Koutsoupias and
                  Angelina Vidali},
  title        = {A Lower Bound of 1+\emph{{\(\varphi\)}} for Truthful Scheduling Mechanisms},
  journal      = {Algorithmica},
  volume       = {66},
  number       = {1},
  pages        = {211--223},
  year         = {2013},
  url          = {https://doi.org/10.1007/s00453-012-9634-6},
  doi          = {10.1007/S00453-012-9634-6},
  bibsource    = {dblp computer science bibliography, https://dblp.org}
}

@inproceedings{DBLP:conf/focs/0001KK21,
  author       = {George Christodoulou and
                  Elias Koutsoupias and
                  Annam{\'{a}}ria Kov{\'{a}}cs},
  title        = {On the Nisan-Ronen conjecture},
  booktitle    = {62nd {IEEE} Annual Symposium on Foundations of Computer Science, {FOCS}
                  2021, Denver, CO, USA, February 7-10, 2022},
  pages        = {839--850},
  publisher    = {{IEEE}},
  year         = {2021},
  url          = {https://doi.org/10.1109/FOCS52979.2021.00086},
  doi          = {10.1109/FOCS52979.2021.00086},
  bibsource    = {dblp computer science bibliography, https://dblp.org}
}

@inproceedings{DBLP:conf/stoc/0001KK23,
  author       = {George Christodoulou and
                  Elias Koutsoupias and
                  Annam{\'{a}}ria Kov{\'{a}}cs},
  editor       = {Barna Saha and
                  Rocco A. Servedio},
  title        = {A Proof of the Nisan-Ronen Conjecture},
  booktitle    = {Proceedings of the 55th Annual {ACM} Symposium on Theory of Computing,
                  {STOC} 2023, Orlando, FL, USA, June 20-23, 2023},
  pages        = {672--685},
  publisher    = {{ACM}},
  year         = {2023},
  url          = {https://doi.org/10.1145/3564246.3585176},
  doi          = {10.1145/3564246.3585176},
  bibsource    = {dblp computer science bibliography, https://dblp.org}
}

@article{DBLP:journals/siamam/Graham69,
  author       = {Ronald L. Graham},
  title        = {Bounds on Multiprocessing Timing Anomalies},
  journal      = {{SIAM} Journal of Applied Mathematics},
  volume       = {17},
  number       = {2},
  pages        = {416--429},
  year         = {1969},
  bibsource    = {dblp computer science bibliography, https://dblp.org}
}

@inproceedings{DBLP:conf/esa/FleischerW00,
  author       = {Rudolf Fleischer and
                  Michaela Wahl},
  editor       = {Mike Paterson},
  title        = {Online Scheduling Revisited},
  booktitle    = {Algorithms - {ESA} 2000, 8th Annual European Symposium, Saarbr{\"{u}}cken,
                  Germany, September 5-8, 2000, Proceedings},
  series       = {Lecture Notes in Computer Science},
  volume       = {1879},
  pages        = {202--210},
  publisher    = {Springer},
  year         = {2000},
  url          = {https://doi.org/10.1007/3-540-45253-2\_19},
  doi          = {10.1007/3-540-45253-2\_19},
  bibsource    = {dblp computer science bibliography, https://dblp.org}
}

@article{DBLP:journals/siamcomp/Albers99,
  author       = {Susanne Albers},
  title        = {Better Bounds for Online Scheduling},
  journal      = {{SIAM} J. Comput.},
  volume       = {29},
  number       = {2},
  pages        = {459--473},
  year         = {1999},
  url          = {https://doi.org/10.1137/S0097539797324874},
  doi          = {10.1137/S0097539797324874},
  bibsource    = {dblp computer science bibliography, https://dblp.org}
}

@article{DBLP:journals/jcss/BartalFKV95,
  author       = {Yair Bartal and
                  Amos Fiat and
                  Howard J. Karloff and
                  Rakesh Vohra},
  title        = {New Algorithms for an Ancient Scheduling Problem},
  journal      = {J. Comput. Syst. Sci.},
  volume       = {51},
  number       = {3},
  pages        = {359--366},
  year         = {1995},
  url          = {https://doi.org/10.1006/jcss.1995.1074},
  doi          = {10.1006/JCSS.1995.1074},
  bibsource    = {dblp computer science bibliography, https://dblp.org}
}

@article{DBLP:journals/jal/KargerPT96,
  author       = {David R. Karger and
                  Steven J. Phillips and
                  Eric Torng},
  title        = {A Better Algorithm for an Ancient Scheduling Problem},
  journal      = {J. Algorithms},
  volume       = {20},
  number       = {2},
  pages        = {400--430},
  year         = {1996},
  url          = {https://doi.org/10.1006/jagm.1996.0019},
  doi          = {10.1006/JAGM.1996.0019},
  bibsource    = {dblp computer science bibliography, https://dblp.org}
}

@article{DBLP:journals/jal/BermanCK00,
  author       = {Piotr Berman and
                  Moses Charikar and
                  Marek Karpinski},
  title        = {On-Line Load Balancing for Related Machines},
  journal      = {J. Algorithms},
  volume       = {35},
  number       = {1},
  pages        = {108--121},
  year         = {2000},
  url          = {https://doi.org/10.1006/jagm.1999.1070},
  doi          = {10.1006/JAGM.1999.1070},
  bibsource    = {dblp computer science bibliography, https://dblp.org}
}

@article{DBLP:journals/jacm/AspnesAFPW97,
  author       = {James Aspnes and
                  Yossi Azar and
                  Amos Fiat and
                  Serge A. Plotkin and
                  Orli Waarts},
  title        = {On-line routing of virtual circuits with applications to load balancing
                  and machine scheduling},
  journal      = {J. {ACM}},
  volume       = {44},
  number       = {3},
  pages        = {486--504},
  year         = {1997},
  url          = {https://doi.org/10.1145/258128.258201},
  doi          = {10.1145/258128.258201},
  bibsource    = {dblp computer science bibliography, https://dblp.org}
}

@article{DBLP:journals/jal/AzarNR95,
  author       = {Yossi Azar and
                  Joseph Naor and
                  Raphael Rom},
  title        = {The Competitiveness of On-Line Assignments},
  journal      = {J. Algorithms},
  volume       = {18},
  number       = {2},
  pages        = {221--237},
  year         = {1995},
  url          = {https://doi.org/10.1006/jagm.1995.1008},
  doi          = {10.1006/JAGM.1995.1008},
  bibsource    = {dblp computer science bibliography, https://dblp.org}
}

@article{DBLP:journals/ipl/BartalKR94,
  author       = {Yair Bartal and
                  Howard J. Karloff and
                  Yuval Rabani},
  title        = {A Better Lower Bound for On-Line Scheduling},
  journal      = {Inf. Process. Lett.},
  volume       = {50},
  number       = {3},
  pages        = {113--116},
  year         = {1994},
  url          = {https://doi.org/10.1016/0020-0190(94)00026-3},
  doi          = {10.1016/0020-0190(94)00026-3},
  bibsource    = {dblp computer science bibliography, https://dblp.org}
}

@article{DBLP:journals/actaC/FaigleKT89,
  author       = {Ulrich Faigle and
                  Walter Kern and
                  Gy{\"{o}}rgy Tur{\'{a}}n},
  title        = {On the performance of on-line algorithms for partition problems},
  journal      = {Acta Cybern.},
  volume       = {9},
  number       = {2},
  pages        = {107--119},
  year         = {1989},
  url          = {https://cyber.bibl.u-szeged.hu/index.php/actcybern/article/view/3359},
  bibsource    = {dblp computer science bibliography, https://dblp.org}
}

@inproceedings{DBLP:conf/soda/GormleyRTW00,
  author       = {Todd Gormley and
                  Nick Reingold and
                  Eric Torng and
                  Jeffery R. Westbrook},
  editor       = {David B. Shmoys},
  title        = {Generating adversaries for request-answer games},
  booktitle    = {Proceedings of the Eleventh Annual {ACM-SIAM} Symposium on Discrete
                  Algorithms, January 9-11, 2000, San Francisco, CA, {USA}},
  pages        = {564--565},
  publisher    = {{ACM/SIAM}},
  year         = {2000},
  url          = {http://dl.acm.org/citation.cfm?id=338219.338608},
  bibsource    = {dblp computer science bibliography, https://dblp.org}
}

@article{DBLP:journals/siamcomp/RudinC03,
  author       = {Rudin III, John F. and
                  Chandrasekaran, R.},
  title        = {Improved Bounds for the Online Scheduling Problem},
  journal      = {{SIAM} J. Comput.},
  volume       = {32},
  number       = {3},
  pages        = {717--735},
  year         = {2003},
  url          = {https://doi.org/10.1137/S0097539702403438},
  doi          = {10.1137/S0097539702403438},
  bibsource    = {dblp computer science bibliography, https://dblp.org}
}

@article{DBLP:journals/algorithmica/AvidorAS01,
  author       = {Adi Avidor and
                  Yossi Azar and
                  Jir{\'{\i}} Sgall},
  title        = {Ancient and New Algorithms for Load Balancing in the \emph{l}\({}_{\mbox{p}}\)
                  Norm},
  journal      = {Algorithmica},
  volume       = {29},
  number       = {3},
  pages        = {422--441},
  year         = {2001},
  url          = {https://doi.org/10.1007/s004530010051},
  doi          = {10.1007/S004530010051},
  bibsource    = {dblp computer science bibliography, https://dblp.org}
}

@inproceedings{DBLP:conf/focs/AwerbuchAGKKV95,
  author       = {Baruch Awerbuch and
                  Yossi Azar and
                  Edward F. Grove and
                  Ming{-}Yang Kao and
                  P. Krishnan and
                  Jeffrey Scott Vitter},
  title        = {Load Balancing in the L\({}_{\mbox{p}}\) Norm},
  booktitle    = {36th Annual Symposium on Foundations of Computer Science, Milwaukee,
                  Wisconsin, USA, 23-25 October 1995},
  pages        = {383--391},
  publisher    = {{IEEE} Computer Society},
  year         = {1995},
  url          = {https://doi.org/10.1109/SFCS.1995.492494},
  doi          = {10.1109/SFCS.1995.492494},
  bibsource    = {dblp computer science bibliography, https://dblp.org}
}

@inproceedings{DBLP:conf/soda/Caragiannis08,
  author       = {Ioannis Caragiannis},
  editor       = {Shang{-}Hua Teng},
  title        = {Better bounds for online load balancing on unrelated machines},
  booktitle    = {Proceedings of the Nineteenth Annual {ACM-SIAM} Symposium on Discrete
                  Algorithms, {SODA} 2008, San Francisco, California, USA, January 20-22,
                  2008},
  pages        = {972--981},
  publisher    = {{SIAM}},
  year         = {2008},
  url          = {http://dl.acm.org/citation.cfm?id=1347082.1347188},
  bibsource    = {dblp computer science bibliography, https://dblp.org}
}

@inproceedings{DBLP:conf/stoc/ImKPS18,
  author       = {Sungjin Im and
                  Nathaniel Kell and
                  Debmalya Panigrahi and
                  Maryam Shadloo},
  editor       = {Ilias Diakonikolas and
                  David Kempe and
                  Monika Henzinger},
  title        = {Online load balancing on related machines},
  booktitle    = {Proceedings of the 50th Annual {ACM} {SIGACT} Symposium on Theory
                  of Computing, {STOC} 2018, Los Angeles, CA, USA, June 25-29, 2018},
  pages        = {30--43},
  publisher    = {{ACM}},
  year         = {2018},
  url          = {https://doi.org/10.1145/3188745.3188966},
  doi          = {10.1145/3188745.3188966},
  bibsource    = {dblp computer science bibliography, https://dblp.org}
}

@inproceedings{DBLP:conf/sigecom/FeldmanFR17,
  author       = {Michal Feldman and
                  Amos Fiat and
                  Alan Roytman},
  editor       = {Constantinos Daskalakis and
                  Moshe Babaioff and
                  Herv{\'{e}} Moulin},
  title        = {Makespan Minimization via Posted Prices},
  booktitle    = {Proceedings of the 2017 {ACM} Conference on Economics and Computation,
                  {EC} '17, Cambridge, MA, USA, June 26-30, 2017},
  pages        = {405--422},
  publisher    = {{ACM}},
  year         = {2017},
  url          = {https://doi.org/10.1145/3033274.3085129},
  doi          = {10.1145/3033274.3085129},
  bibsource    = {dblp computer science bibliography, https://dblp.org}
}

@article{DBLP:journals/scheduling/LiLW23,
  author       = {Bo Li and
                  Minming Li and
                  Xiaowei Wu},
  title        = {Well-behaved online load balancing against strategic jobs},
  journal      = {J. Sched.},
  volume       = {26},
  number       = {5},
  pages        = {443--455},
  year         = {2023},
  url          = {https://doi.org/10.1007/s10951-022-00770-6},
  doi          = {10.1007/S10951-022-00770-6},
  bibsource    = {dblp computer science bibliography, https://dblp.org}
}

@inproceedings{DBLP:conf/soda/AlonAWY97,
  author       = {Noga Alon and
                  Yossi Azar and
                  Gerhard J. Woeginger and
                  Tal Yadid},
  editor       = {Michael E. Saks},
  title        = {Approximation Schemes for Scheduling},
  booktitle    = {Proceedings of the Eighth Annual {ACM-SIAM} Symposium on Discrete
                  Algorithms, 5-7 January 1997, New Orleans, Louisiana, {USA}},
  pages        = {493--500},
  publisher    = {{ACM/SIAM}},
  year         = {1997},
  url          = {http://dl.acm.org/citation.cfm?id=314161.314371},
  bibsource    = {dblp computer science bibliography, https://dblp.org}
}

@article{DBLP:journals/algorithmica/EpsteinS04,
  author       = {Leah Epstein and
                  Jir{\'{\i}} Sgall},
  title        = {Approximation Schemes for Scheduling on Uniformly Related and Identical
                  Parallel Machines},
  journal      = {Algorithmica},
  volume       = {39},
  number       = {1},
  pages        = {43--57},
  year         = {2004},
  url          = {https://doi.org/10.1007/s00453-003-1077-7},
  doi          = {10.1007/S00453-003-1077-7},
  bibsource    = {dblp computer science bibliography, https://dblp.org}
}

@article{DBLP:journals/mp/LenstraST90,
  author       = {Jan Karel Lenstra and
                  David B. Shmoys and
                  {\'{E}}va Tardos},
  title        = {Approximation Algorithms for Scheduling Unrelated Parallel Machines},
  journal      = {Math. Program.},
  volume       = {46},
  pages        = {259--271},
  year         = {1990},
  url          = {https://doi.org/10.1007/BF01585745},
  doi          = {10.1007/BF01585745},
  bibsource    = {dblp computer science bibliography, https://dblp.org}
}

@article{DBLP:journals/disopt/DosaT10,
  author       = {Gy{\"{o}}rgy D{\'{o}}sa and
                  Zhiyi Tan},
  title        = {New upper and lower bounds for online scheduling with machine cost},
  journal      = {Discret. Optim.},
  volume       = {7},
  number       = {3},
  pages        = {125--135},
  year         = {2010},
  url          = {https://doi.org/10.1016/j.disopt.2010.02.005},
  doi          = {10.1016/J.DISOPT.2010.02.005},
  bibsource    = {dblp computer science bibliography, https://dblp.org}
}

@article{DBLP:journals/orl/Tichy04,
  author       = {Tom{\'{a}}s Tich{\'{y}}},
  title        = {Randomized on-line scheduling on three processors},
  journal      = {Oper. Res. Lett.},
  volume       = {32},
  number       = {2},
  pages        = {152--158},
  year         = {2004},
  url          = {https://doi.org/10.1016/j.orl.2003.05.003},
  doi          = {10.1016/J.ORL.2003.05.003},
  bibsource    = {dblp computer science bibliography, https://dblp.org}
}

@article{DBLP:journals/ipl/Sgall97,
  author       = {Jir{\'{\i}} Sgall},
  title        = {A Lower Bound for Randomized On-Line Multiprocessor Scheduling},
  journal      = {Inf. Process. Lett.},
  volume       = {63},
  number       = {1},
  pages        = {51--55},
  year         = {1997},
  url          = {https://doi.org/10.1016/S0020-0190(97)00093-8},
  doi          = {10.1016/S0020-0190(97)00093-8},
  bibsource    = {dblp computer science bibliography, https://dblp.org}
}

@article{DBLP:journals/algorithmica/Seiden00,
  author       = {Steven S. Seiden},
  title        = {Online Randomized Multiprocessor Scheduling},
  journal      = {Algorithmica},
  volume       = {28},
  number       = {2},
  pages        = {173--216},
  year         = {2000},
  url          = {https://doi.org/10.1007/s004530010014},
  doi          = {10.1007/S004530010014},
  bibsource    = {dblp computer science bibliography, https://dblp.org}
}

@inproceedings{DBLP:conf/stoc/Albers02,
  author       = {Susanne Albers},
  editor       = {John H. Reif},
  title        = {On randomized online scheduling},
  booktitle    = {Proceedings on 34th Annual {ACM} Symposium on Theory of Computing,
                  May 19-21, 2002, Montr{\'{e}}al, Qu{\'{e}}bec, Canada},
  pages        = {134--143},
  publisher    = {{ACM}},
  year         = {2002},
  url          = {https://doi.org/10.1145/509907.509930},
  doi          = {10.1145/509907.509930},
  bibsource    = {dblp computer science bibliography, https://dblp.org}
}

@inproceedings{DBLP:conf/stoc/AspnesAFPW93,
  author       = {James Aspnes and
                  Yossi Azar and
                  Amos Fiat and
                  Serge A. Plotkin and
                  Orli Waarts},
  editor       = {S. Rao Kosaraju and
                  David S. Johnson and
                  Alok Aggarwal},
  title        = {On-line load balancing with applications to machine scheduling and
                  virtual circuit routing},
  booktitle    = {Proceedings of the Twenty-Fifth Annual {ACM} Symposium on Theory of
                  Computing, May 16-18, 1993, San Diego, CA, {USA}},
  pages        = {623--631},
  publisher    = {{ACM}},
  year         = {1993},
  url          = {https://doi.org/10.1145/167088.167248},
  doi          = {10.1145/167088.167248},
  bibsource    = {dblp computer science bibliography, https://dblp.org}
}

@article{DBLP:journals/mst/EbenlendrS15,
  author       = {Tom{\'{a}}s Ebenlendr and
                  Jir{\'{\i}} Sgall},
  title        = {A Lower Bound on Deterministic Online Algorithms for Scheduling on
                  Related Machines Without Preemption},
  journal      = {Theory Comput. Syst.},
  volume       = {56},
  number       = {1},
  pages        = {73--81},
  year         = {2015},
  url          = {https://doi.org/10.1007/s00224-013-9451-6},
  doi          = {10.1007/S00224-013-9451-6},
  bibsource    = {dblp computer science bibliography, https://dblp.org}
}

@article{DBLP:journals/orl/EpsteinS00,
  author       = {Leah Epstein and
                  Jir{\'{\i}} Sgall},
  title        = {A lower bound for on-line scheduling on uniformly related machines},
  journal      = {Oper. Res. Lett.},
  volume       = {26},
  number       = {1},
  pages        = {17--22},
  year         = {2000},
  url          = {https://doi.org/10.1016/S0167-6377(99)00062-0},
  doi          = {10.1016/S0167-6377(99)00062-0},
  bibsource    = {dblp computer science bibliography, https://dblp.org}
}

\newpage
\appendix
\section{Online Selfish Scheduling on Unrelated Machines}\label{sec:unrelated}

In this section, we discuss online selfish scheduling on unrelated machines, where each selfish machine~$i$ has private processing times~$p_{ij}$ for each arriving job~$j$ and reports its processing time upon the arrival of each job.  
The goal of each machine is to maximize its utility, defined as the total payment minus the total completion time.  
We show that the mechanism proposed by~\citet{DBLP:conf/stoc/NisanR99} remains truthful even in the online setting.  
Moreover, when the objective is the makespan, its competitive ratio remains~$m$, which is tight as shown by~\citet{DBLP:conf/stoc/0001KK23}; when the objective is the~$\ell_q$ norm, the competitive ratio is~$m^{1-1/q}$.

\begin{algorithm}[tb]
\caption{Truthful Mechanism on Unrelated Machines~\cite{DBLP:conf/stoc/NisanR99}}
\label{alg:unrelated}
\KwIn{The online reported processing time $\vp= \{p_{ij}\}_{i\in [m],j\in [n]}$.}
\KwOut{ An online allocation $\vX = \{x_{ij}\}_{i\in [m],j\in [n]}$ and payment $\vP =  \{P_{ij}\}_{i\in [m],j\in [n]}$. }

\For{each arriving job $j$ }
{
Let \(i\) and \(i'\) denote the machines with the shortest and second shortest processing times, respectively\;

Set \(x_{ij} \gets 1\) and \(P_{ij} \gets p_{i'j}\); set \(x_{kj}\) and \(P_{kj}\) to~0 for all \(k \neq i\).

}
\Return{$\vX = \{x_{ij}\}_{i\in [m],j\in [n]}$ and $\vP = \{P_{ij}\}_{i\in [m],j\in [n]}$.}

\end{algorithm}

\begin{theorem}
  \Cref{alg:unrelated} is truthful for selfish machines and achieves a competitive ratio of~$m$ for the makespan objective and~$m^{1-1/q}$ for the~$\ell_q$ norm objective.
\end{theorem}

\begin{proof}
The truthfulness and the makespan competitive ratio analysis are almost identical to the analysis given in~\cite{DBLP:conf/stoc/NisanR99}, because the entire mechanism considers each job independently.  
Whether or not the machines know future jobs, their behavior remains the same as in the offline setting.  

For the $\ell_q$ norm objective, we observe that assigning each job to the machine with the shortest processing time always yields the optimal $\ell_1$ norm objective.  
By H\"older's inequality, for any \(q > 1\), the ratio of a vector's $\ell_1$ norm to its $\ell_q$ norm lies between~1 and~\(m^{1-1/q}\).  
Therefore, the solution produced by our mechanism achieves a competitive ratio of at most~\(m^{1-1/q}\) for the $\ell_q$ norm objective.
\end{proof}


\section{Analysis of The $\ell_q$ Norm Objective}
\label{sec:app-norm}
\label{sec:lq-analysis}
\begin{algorithm}[htbp]
\caption{Level-Based Proportional Allocation for $\ell_q$ Norm}
\label{alg:lq}
\KwIn{The reported speeds $\{s_i\}_{i\in [m]}$ and the reported job sizes $\{p_j\}_{j\in [n]}$ which shows up online.}
\KwOut{ An online fractional allocation $\vX = \{x_{ij}\}_{i\in [m],j\in [n]}$ }
Sort $\{s_i\}_{i\in [m]}$ in descending order \;
$\forall i\in [m]$, set $\bs_i \gets \max \{ 2^z \mid 2^z\leq s_i,~z\in \Z\}$ \tcp*{round down $s_i$ to the nearest power of $2$}

Ignore the machines with $\bs_i < \bs_1/m$ and partition the remaining machines into $K=\lfloor \log m \rfloor+1$ groups: $\forall 1\leq k \leq K$, let $\cM_k \gets \{ i\in [m] \mid \bs_i = r_k \}$ and $\cpM{k} \gets \bigcup_{k'\leq k} \cM_k$, where $r_k:= \frac{\bs_1}{2^{k-1}}$ \;
\For{each arriving job $j$ }
{
\eIf{$j=1$}
{
$\forall i\in \cM_1$, set $x_{ij}\gets \frac{1}{|\cM_1|}$ \;

Set $\Lambda \gets p_1/r_1$  \tcp*{initialize a guessed optimal objective}

Set $C_{i,k} \gets 0$ for all $k \in [K]$ and $i \in \cpM{k}$ \tcp*{$K=\lfloor \log m \rfloor+1$}

}
{Define $\kj \gets \max \bigg(\{k \in [K]  \mid p_j \leq r_k \cdot \Lambda \} \cup \{1\}\bigg) $ \tcp*{decide job $j$'s level}

\For{each $i\in \cpM{\kj}$}
{
    $x_{ij} \leftarrow \frac{\bs_i^{\gamma}}{\sum_{i'\in \cpM{\kj}} \bs_{i'}^{\gamma}}$ \tcp*{$\gamma = \frac{q}{q-1}$} 
    
    $C_{i,\kj} \leftarrow C_{i,\kj} + x_{ij}\cdot \frac{p_j}{\bs_i}$ \; 
}
\If(\tcp*[f]{double-without-the-last}){ $k(j) \leq K-1$ } 
{
\If(\tcp*[f]{a super large job}){$p_j / r_1 > \Lambda $ }{Keep doubling $\Lambda$ until it is at least $p_j/r_1$ \; Reset $C_{i,k} \gets 0$ for all $k \in [K]$ and $i \in \cpM{k}$ \tcp*{start a new phase}}
\ElseIf(\tcp*[f]{a saturated level}){ $ || \vC_{\kj} ||_q > \Lambda$ }{ Double $\Lambda$ \; Reset $C_{i,k} \gets 0$ for all $k \in [K]$ and $i \in \cpM{k}$ \tcp*{start a new phase}}
}
}
}

\Return{$\vX = \{x_{ij}\}_{i\in [m],j\in [n]}$.}
\end{algorithm}

\subsection{Two-sided Montonicity for \Cref{alg:lq}}
\label{sec:norm-monotone}

This subsection proves the two-sided monotonicity of~\cref{alg:lq}. 

\begin{lemma}[Two-Sided Monotonicity, Corresponding to~\cref{lem:makespan:monontone}] \label{lem:lq:monontone}
The allocation returned by~\cref{alg:lq} is both machine-side and job-side monotone.
\end{lemma}

\begin{proof}
    We still consider the scenario that machine $i^*$ (without loss of generality, we assume $i^*$ is the smallest index in its group) increases its speed from $\bs_{i^*}$ to $\bs_{i^*}'=2\bs_{i^*}$.
    The goal is to show that 
    \begin{equation}
        \label{eqn:lq:each-job-increase}
        \forall j\in [n], \quad x_{i^*j} \leq x_{i^*j}'~.
    \end{equation}
     The inequality trivially holds for the first arriving job or the jobs with $x_{i^*j}=0$. For each remaining job $j$ with $x_{i^*j}>0$, we know that $i^* \in \cpM{\kj}$. Therefore, to prove that $x_{i^*j} \leq x_{i^*j}'$, we must first establish that after the speed increases, machine $i^*$ still remains part of the feasible machine set $\cpM{k'(j)}'$ of job $j$ (i.e., $i^*\in \cpM{k'(j)}'$); otherwise, $x_{i^*j}'$ would be $0$. Furthermore, according to the proportional allocation rule of the algorithm,  
     \[x_{i^*j} \leftarrow \frac{\bs_{i^*}^{\gamma}}{\sum_{t\in \cpM{\kj}} \bs_{t}^{\gamma}}\]
     and 
     \[x_{i^*j}' \leftarrow \frac{\bs_{i^*}'^{\gamma}}{\sum_{t\in \cpM{k'(j)}'} \bs_{t}'^{\gamma}}~.\]
     Since $\bs_{i^*} < \bs_{i^*}'$, if we can demonstrate that $\cpM{k'(j)}' \subseteq \cpM{\kj}$, which means the feasible machine set for job $j$ does not gain any new machine due to the speed increase of machine $i^*$, then $x_{i^*j} \leq x_{i^*j}'$ can be proved. In summary, we need the following two properties for each job $j$ to prove~\cref{eqn:lq:each-job-increase}:
    \begin{enumerate}[left=2em]
        \item[(\rom{1})]  If $i^*\in \cpM{\kj}$, then $i^*\in \cpM{k'(j)}'$~.
        \item[(\rom{2})] $\cpM{k'(j)}' \subseteq \cpM{\kj}$~.
    \end{enumerate}

    Similarly, these two properties can be proved straightforward by the $\Lambda$ Stability (\cref{lem:lq:lambda_stable}). 
\end{proof}

We employ the same definition of a good job $t$: 
\begin{itemize}[left=2em]
    \item The guessed $\opt$ is the same in both worlds, i.e., $\Lambda_t=\Lambda_t'$.
    \item The current completion time (every level) of the speed-up world  $C_{i,k}'=0$ for all $k \in [K]$ and $i \in \cpM{k}$ at the arrival time of job $t$ (before allocating $t$).
\end{itemize}

\begin{lemma}[Corresponding to \cref{lem:equaltime}]
\label{lem:lq:equaltime}
    Let $t$ be a good job (time), and $j$ be the first time $\Lambda_j' > \Lambda_t'$. We have $\Lambda_j > \Lambda_t$. 
\end{lemma}
\begin{proof}

The proof is similar to the proof in makespan minimization. We also consider the contradiction that $\Lambda_j = \Lambda_t$ and can prove that the increase of $\Lambda_j'$ right after allocating job $j-1$ is not due to the arrival of a super large job. Then
$\Lambda'$ should increase due to a saturated level: $||\vC_{k'(j-1)}'||_q > \Lambda_{j-1}'$. Again, it seems impossible that after allocating job $j-1$, $\Lambda'$ is doubled due to $||\vC_{k'(j-1)}' ||_q> \Lambda_{j-1}'$ while $\Lambda$ has not been doubled, as machines in the speed-up world are more powerful. By carefully analyzing the allocation of each job $h\in [t,j-1]$, we can finally draw a contradiction by showing that $\Lambda$ must be doubled right after allocating job $j-1$ due to \[||\vC_{k(j-1)}||_q\geq ||\vC_{k'(j-1)}'||_q > \Lambda_{j-1}'=\Lambda_{j-1}~,\]
    which completes the proof.
    
We also want to construct job-level consistency and job-$\ell_q$-norm monotonicity properties to prove that the machine power in the speed-up world is stronger. The proof of job-level consistency property is the same as that in the makespan setting. The difference happens in the job-$\ell_q$-norm monotonicity. 

    \begin{property}[Job-Level Consistency, Corresponding to~\cref{pro:makespan:job-consist}]\label{pro:lq:job-consist}
         We call a job $h$ \emph{normal} if \\ $\max\{k(h), k'(h)\} \leq K-1$ and $p_h \leq \min \{ r_1 \Lambda_h,~ r_1' \Lambda_h'\}$. For any two \emph{normal} jobs $h_1$ and $h_2$ with $ \Lambda_{h_1}=\Lambda_{h_2} $ and $\Lambda_{h_1}' = \Lambda_{h_2}'$, we have $k(h_1)=k(h_2)$ if and only if $k'(h_1)=k'(h_2)$. 
    \end{property}

    \begin{property}[Job-$\ell_q$-Norm Monotonicity, Corresponding to~\cref{pro:makespan:job-mono}]\label{pro:lq:job-mono}
        For each job $h$ with $\Lambda_h=\Lambda_h'$, when $k'(h)\leq K-1$ (not the last level),  $\displaystyle \left({\sum_{i\in \cpM{k(h)}}\left(\frac{x_{ik(h)}}{\bs_i}\right)^q}\right)^{\frac{1}{q}}\geq \left({\sum_{i\in \cpM{k'(h)}'}\left(\frac{x_{ik'(h)}}{\bs_i'}\right)^q}\right)^{\frac{1}{q}}$. 
    \end{property}

The proof of~\cref{pro:lq:job-consist} is verbatim as that of~\cref{pro:makespan:job-consist} as it is unrelated to the objective. For the proof of~\cref{pro:lq:job-mono}, we defer it to later and assume it is correct to complete the overall argument.
Let us define $J=\{h\in [t,j-1]\mid k(h)=k(j-1)\}$ and $J'=\{h \in [t,j-1] \mid k'(h)=k'(j-1)\}$. 
In addition, in the $\ell_q$ norm environment, we define $U=\{h\leq j-1\mid k(h)=k(j-1) \text{ and }  \Lambda_h = \Lambda_{j-1}\}$ and $U'=\{h \leq j-1 \mid k'(h)=k'(j-1)  \text{ and } \Lambda_h' = \Lambda_{j-1}' \}$.
As job $j-1$ triggers doubling and we use the double-without-the-last trick\footnote{Note that if the algorithm does not use the double-without-the-last trick, we cannot prove that job $j-1$ is normal, which results in the failure to apply~\cref{pro:makespan:job-consist}. }, we have $k'(j-1)\leq K-1$ and therefore, $k'(h)\leq K-1$ for any $ h\in J'$. Additionally, due to~\cref{eqn:contradict} and $r_i\leq r_i'$ for each $i$, we have $k(h) \leq k'(h)\leq K-1$ for any $h\in J'$.
Then, since the algorithm never performs doubling due to a super large job in both worlds, we can conclude that all jobs in $J'$ are normal. By~\cref{pro:lq:job-consist}, for each job $h\in J'$, we have $k(h) = k(j-1)$ since $k'(h) = k'(j-1)$, and therefore, $J'\subseteq J$.
Finally, combining with~\Cref{pro:lq:job-mono}, right after job ${j-1}$ is allocated, we have 

\begin{align*}
        &||\vC_{k(j-1)}||_q - ||\vC_{k'(j-1)}'||_q \\
        &=\left({\sum_{i\in \cpM{k(j-1)}}C_{i,k(j-1)}^q}\right)^{\frac{1}{q}} - \left({\sum_{i\in \cpM{k'(j-1)}}C_{i,k'(j-1)}^q}\right)^{\frac{1}{q}} \\
        &=\left({\sum_{i\in \cpM{k(j-1)}}\left(\sum_{h\in U}x_{ih}\cdot \frac{p_{h}}{\bs_i}\right)^q}\right)^{\frac{1}{q}} - \left({\sum_{i\in \cpM{k'(j-1)}'}\left(\sum_{h\in U'}x_{ih}'\cdot \frac{p_{h}}{\bs_i'}\right)^q}\right)^{\frac{1}{q}}\\
        &=\sum_{h\in U}p_{h}\cdot\left({\sum_{i\in \cpM{k(j-1)}}\left(\frac{x_{ik(j-1)}}{\bs_i}\right)^q}\right)^{\frac{1}{q}}-\sum_{h\in U'}p_{h}\cdot\left({\sum_{i\in \cpM{k'(j-1)}'}\left(\frac{x'_{ik'(j-1)}}{\bs_i'}\right)^q}\right)^{\frac{1}{q}}\\
        & \geq \sum_{h\in J}p_{h}\cdot\left({\sum_{i\in \cpM{k(j-1)}}\left(\frac{x_{ik(j-1)}}{\bs_i}\right)^q}\right)^{\frac{1}{q}}-\sum_{h\in J'}p_{h}\cdot\left({\sum_{i\in \cpM{k'(j-1)}'}\left(\frac{x'_{ik'(j-1)}}{\bs_i'}\right)^q}\right)^{\frac{1}{q}} \tag{the second condition of good jobs}\\
        & \geq \sum_{h\in J'}p_{h}\cdot\left({\sum_{i\in \cpM{k(j-1)}}\left(\frac{x_{ik(j-1)}}{\bs_i}\right)^q}\right)^{\frac{1}{q}}-\sum_{h\in J'}p_{h}\cdot\left({\sum_{i\in \cpM{k'(j-1)}'}\left(\frac{x'_{ik'(j-1)}}{\bs_i'}\right)^q}\right)^{\frac{1}{q}} \tag{$J'\subseteq J$} \\
        & \geq 0~. \tag{\cref{pro:lq:job-mono}}
    \end{align*}
 Notice that the first inequality can be proved by the second condition of good jobs because at the arrival time of job $t$, the corresponding $||\vC_{k(j-1)}||_q $ and $||\vC_{k(j-1)}'||_q $ are equal to $\displaystyle \sum_{h\in U\setminus J}p_{h}\cdot\left({\sum_{i\in \cpM{k(j-1)}}\left(\frac{x_{ik(j-1)}}{\bs_i}\right)^q}\right)^{\frac{1}{q}}$ and $\displaystyle \sum_{h\in U'\setminus J'}p_{h}\cdot\left({\sum_{i\in \cpM{k'(j-1)}'}\left(\frac{x'_{ik'(j-1)}}{\bs_i'}\right)^q}\right)^{\frac{1}{q}}$, respectively.
\end{proof}

\begin{proof}[Proof of~\cref{pro:lq:job-mono}]
    
    According to our proportional allocation rule, we have
    \begin{align*}
         \left({\sum_{i\in \cpM{k(h)}}\left(\frac{x_{ik(h)}}{\bs_i}\right)^q}\right)^{\frac{1}{q}}&=\left(\sum_{i\in \cpM{k(h)}}\left(\frac{\bs_i^{\gamma}}{\sum_{i'\in \cpM{k(h)}} \bs_{i'}^{\gamma}}\cdot \frac{1}{\bs_i}\right)^q\right)^{\frac{1}{q}}\\
         &=\left(\sum_{i\in \cpM{k(h)}}\bs_{i}^\gamma\right)^{-\frac{1}{\gamma}}\tag{$\gamma=\frac{q}{q-1}$ }~.
    \end{align*}
    Similarly, we have $\left({\sum_{i\in \cpM{k'(h)}'}\left(\frac{x_{ik'(h)}}{\bs_i'}\right)^q}\right)^{\frac{1}{q}}=\left(\sum_{i\in \cpM{k'(h)}'}\bs_{i}'^{\gamma}\right)^{-\frac{1}{\gamma}}$.
    
    For each job $h$ with $\Lambda_h=\Lambda_h'$ and $k'(h) \leq K-1$, the set of feasible machines will only increase with machine acceleration, i.e., $\cpM{k(h)} \subseteq \cpM{k'(h)}'$.
    We have 
    \begin{align*}
        \left(\sum_{i\in \cpM{k(h)}}\bs_{i}^\gamma\right)^{\frac{1}{\gamma}}&\leq \left(\sum_{i\in \cpM{k'(h)}'}\bs_{i}^\gamma\right)^{\frac{1}{\gamma}}\tag{$\cpM{k(h)}\subseteq \cpM{k'(h)}'$ }\\
        &\leq \left(\sum_{i\in \cpM{k'(h)}'}\bs_{i}'^{\gamma}\right)^{\frac{1}{\gamma}}~.\\
    \end{align*}
    Therefore, $\left({\sum_{i\in \cpM{k(h)}}\left(\frac{x_{ik(h)}}{\bs_i}\right)^q}\right)^{\frac{1}{q}}\geq \left({\sum_{i\in \cpM{k'(h)}'}\left(\frac{x_{ik'(h)}}{\bs_i'}\right)^q}\right)^{\frac{1}{q}}$.
\end{proof}

\begin{lemma}[$\Lambda$-Stability, Corresponding to~\cref{lem:makespan:lambda_stable}]\label{lem:lq:lambda_stable}
    For any job $j\geq 2$, we have 
    \[\Lambda_j \geq \Lambda_j' \geq \frac{1}{2}\Lambda_j~.\]
\end{lemma}

\begin{proof}
  Building on~\cref{lem:lq:equaltime}, the $\Lambda$-Stability can be proved similarly by finding a good job $t$. The reason why we can find a good job is verbatim as that in~\cref{lem:makespan:lambda_stable}.  
\end{proof}

\subsection{Competitive Analysis}
\label{sec:norm-ratio}



Similar to the mechanism for makespan, our final mechanism executes~\cref{alg:lq} and performs an online rounding to obtain an integral solution: for each job $j$, we independently assign it to machine $i$ with probability $x_{ij}$. Again, this independent rounding does not impact the truthfulness, as the expected utility for each machine or job remains unchanged.

Now, we analyze the competitive ratio. Use $\vY=\{y_{ij}\in \{0,1\}\}$ to denote the random allocation by our final mechanism. Let $\obj(\vX)$, $\obj(\vY)$ and $\opt$ be the objectives corresponding to $\vX$, $\vY$, and the optimal solution, respectively.
We show the following:

\begin{lemma} [Competitive Ratio] \label{thm:lq_ratio}
    With probability at least $1-1/m^2$, $\obj(\vY) \leq O(m^{\frac{1}{q}(1-\frac{1}{q})} \cdot (\log m)^{1+\frac{1}{q^2}} ) \cdot \opt$. Further, the expected competitive ratio of $\vY$ is also guaranteed to $ O(m^{\frac{1}{q}(1-\frac{1}{q})} \cdot (\log m)^{1+\frac{1}{q^2}} )$.
\end{lemma}




Similar to the proof structure of~\cref{thm:makespan_ratio}, we first give a rounding lemma, and then analyze the properties of the fractional solution returned by~\cref{alg:lq}.  

\begin{lemma}[$\ell_q$-Norm Approximately Roundability]\label{lem:lq:rounding}
   For any (online) fractional allocation $\vX$ that is speed-size-feasible, if $\obj(\vX) \leq \alpha \cdot \opt$, then natural independent rounding yields a competitive ratio of $O(m^{\frac{1}{q}(1-\frac{1}{q})} \cdot (\log m)^{\frac{1}{q^2}} \cdot \max\{\log m, \alpha\})$ with probability at least $1-O(1/m^2)$. Furthermore, the rounding is guaranteed to be $O(m^2)$-competitive in the worst case.  
\end{lemma}
\begin{proof}

    The proof also relies on concentration inequalities and the union bound, but compared with the makespan, the $\ell_q$-norm objective requires some additional tricks. 
    Use $\vY=\{y_{ij}\in \{0,1\}\}$ to denote the random allocation obtained through independent rounding. Similar to the proof of~\cref{thm:makespan_ratio}, we have for each machine $i$ and any $\delta \geq 2$,
    \begin{equation}\label{eq:lq:singlemachine}
        \Pr\left[ \sum_{j\in \cJ} y_{ij}\cdot \frac{p_j}{s_i} \geq (1+\delta) \cdot \alpha \cdot \opt \right] \leq \frac{1}{m^3} \cdot \exp{\left(-\frac{\delta\alpha}{2} + 3 \log m\right)}
    \end{equation}
    where we assume w.l.o.g. that the constant in the speed-size feasibility is $1$.

    If we then directly use the union bound and compute the $\ell_q$ norm, we can only derive a competitive ratio of $O(m^{\frac{1}{q}} \cdot \max\{\log m, \alpha \})$. To achieve a better ratio, we employ a more careful analysis.   
    By the mild assumption stated in the lemma, after ignoring the empty machines, the remaining machine can be partitioned into $O(\log m)$ groups $\cG_1,\cG_2,...,\cG_H$ such that any two machines~$a,b$ in one group $\cG_h$ satisfies $\frac{s_a}{s_b} \in [\frac{1}{2},2]$. In the following, we show that with high probability, the contribution of each group is efficiently bounded. 

    Consider an arbitrary group $\cG_h$ with $m_h$ machines. 
    Use $s_{min}$ and $s_{max}$ to denote the minimum and maximum speeds in this group, respectively.
    For each machine $i\in \cG_h$, let $L_i:=\sum_{j\in \cJ} x_{ij}\cdot p_j $ and $R_i := \sum_{j\in \cJ} y_{ij}\cdot p_j  $ be the fractional load and randomized load on it, respectively. Using the condition that $\vX$ is $\alpha$-competitive and the Jensen inequality, we can bound the $\ell_1$-norm of a group's completion times:
    \begin{align*}
        \sum_{i\in \cG_h} \frac{L_i}{s_i} & \leq \frac{1}{s_{min}} \cdot \sum_{i\in \cG_h} L_i \\ 
        & \leq \frac{1}{s_{min}} \cdot m_h^{1-\frac{1}{q}} \cdot  \left( \sum_{i\in \cG_h} L_i ^ q  \right)^{\frac{1}{q}} \tag{Jensen inequality} \\
        & \leq \frac{s_{max}}{s_{min}} \cdot m_h^{1-\frac{1}{q}}\cdot \left( \sum_{i\in \cG_h} \left(\frac{L_i}{s_i}\right)^{q}  \right)^{\frac{1}{q}} \\
        & \leq 2 \cdot m_h^{1-\frac{1}{q}}\cdot \alpha \cdot \opt .
    \end{align*}
    Then, applying Chernoff bound to the $\ell_1$-norm of group $\cG_h$, we have for any $\epsilon \geq 2$,

    \begin{equation}\label{eq:lq:group}
        \Pr \left[ \sum_{i\in \cG_h} \frac{R_i}{s_i} \geq 2 (1+\epsilon)\cdot m_h^{1-\frac{1}{q}}\cdot \alpha \cdot \opt \right] \leq \frac{1}{m^3}\exp{\left(-\epsilon \cdot  m_h^{1-\frac{1}{q}}\cdot \alpha + 3\log m\right)}. 
    \end{equation}
    
    Combing~\cref{eq:lq:singlemachine} and~\cref{eq:lq:group}, we have with probability at least $1-O(1/m^2)$, a good event happens, where any machine $i$ satisfies 
    \[ \frac{R_i}{s_i} \leq O(\max\{\log m, \alpha\}) \cdot \opt, \]
    and any group $\cG_h$ satisfies
    \[ \sum_{i\in \cG_h} \frac{R_i}{s_i} \leq O\left(m_h^{1-\frac{1}{q}} \cdot \max\{\log m, \alpha\}\right) \cdot \opt. \]

    Now, we prove an upper bound of a single group's $\ell_q$-norm objective when a good event occurs. By analyzing the partial derivative of the $\ell_q$ norm function, we have
    \[ \sum_{i\in \cG_h} \left(\frac{R_i}{s_i}\right)^q \leq  m_h^{1-\frac{1}{q}} \cdot \left(O(\max\{\log m, \alpha\}) \cdot \opt\right)^q, \]
    where $c$ is a constant. Summing over all groups, we have:
    \begin{align*}
        \obj(\vY) &= \left(\sum_{h\in [H]} \sum_{i\in \cG_h} \left(\frac{R_i}{s_i}\right)^q \right)^{\frac{1}{q}}\\
        & \leq O(\max\{\log m, \alpha\}) \cdot \opt \cdot  \left(\sum_{h\in [H]} m_h^{1-\frac{1}{q}}\right)^{\frac{1}{q}} \\
        & \leq O(\max\{\log m, \alpha\}) \cdot \opt \cdot H^{\frac{1}{q^2}} \cdot m^{\frac{1}{q}(1-\frac{1}{q})} \tag{Jensen inequality} \\
        & = O(m^{\frac{1}{q}(1-\frac{1}{q})} \cdot (\log m)^{\frac{1}{q^2}} \cdot \max\{\log m, \alpha\}) \cdot \opt \tag{$H=O(\log m)$}
    \end{align*}

    The final piece is similar to the proof in~\cref{lem:rounding}. As $\vX$ is speed-size-feasible, the objective is at most $O(m^2)\cdot \opt$ in the worst case.  
\end{proof}

Now we show that the solution returned by~\cref{alg:lq} satisfies the conditions in the lemma above by~\cref{lem:lq:roundable} and~\cref{lem:lq:fractionalcompetitive}.

\begin{lemma}\label{lem:lq:roundable}
    The fractional solution $\vX$ returned by~\cref{alg:lq} is speed-size-feasible.
\end{lemma}
\begin{proof}
    This proof is similar to the proof of~\cref{lem:makespan:roundable}. The initial scaling and ignoring operations of the algorithm guarantee the machine-speed constraints and can be considered as creating a new machine set $\bcM$ with $\bopt \leq 4\cdot \opt$. Then, each machine in $\cM_1$ must satisfy the job-size constraint as any job $j\in \cJ$ has $p_j / \bs_1 \leq \bopt$. For each other machine $i\notin \cM_1$, according to our job-level criteria, the assigned job's size is bounded by $\bs_i \cdot \Lambda$. Thus, the remaining piece is to demonstrate that the final guess $\Lambda_n$ is at most $2\cdot \bopt$.
    
    

    We assume for contradiction that $\Lambda_n > 2\cdot \bopt$. Find the last iteration $j$ with $\Lambda_j \leq 2\cdot \bopt$. Clearly, the $\Lambda$ increase at the end of iteration $j$ must be due to $||\vC_{\kj}||_q > \Lambda_j$; otherwise, $\Lambda_{j+1}$ cannot be greater than $2\cdot \bopt$.
    Thus, we have $\Lambda_j = \Lambda_{j+1}/2 \geq \bopt$. 
    
    Use $\cJ (\kj) := \{h \in [n] \mid k(h)=\kj \text{ and } \Lambda_h = \Lambda_j \}$ to denote the set of jobs that have the same level and corresponding $\Lambda$ as job $j$.
    Our algorithm distributes these jobs proportionally across $\cpM{\kj}$. According to the definition of a job's level, each job $j'\in \cJ (\kj)$ satisfies $p_{j'} / \bs_i > \Lambda_j \geq\bopt$ for any machine $i\notin \cpM{\kj}$, implying that the optimal solution must allocate job $j'$ to a machine in $\cpM{\kj}$. Let $\pi(i)$ be the set of jobs in $\cJ (\kj)$ that are assigned to machine $i$ in the optimal solution. According to the conclusion in~\cite{DBLP:conf/stoc/ImKPS18} that our proportional allocation rule yields an optimal fractional assignment for each specific level, we have
    \[ ||\vC_{\kj}||_q \leq \left(\sum_{i\in \cpM{\kj}} \left(\sum_{j \in \pi(i)} p_j \right)^{q} \right)^{\frac{1}{q}} \leq \bopt \leq \Lambda_j,  \]
    which contradicts the fact that $\Lambda$ increases due to $||\vC_{\kj}||_q > \Lambda_j$ and completes the proof.
\end{proof}

\begin{lemma}\label{lem:lq:fractionalcompetitive}
    The fractional solution $\vX$ returned by~\cref{alg:lq} has $\obj(\vX) \leq O(\log m) \cdot \opt$.
    
\end{lemma}

\begin{proof}
    In this proof, we also assume that the algorithm assigns jobs to the scaled machine set $\bcM$ with $\bopt \leq 4\cdot \opt$.
    The basic proof idea is to first split $\obj(\vX)$ into $O(\log m)$ parts and then show each part is at most a constant factor of $\bopt$.
    Let $\Lambda^{(t)}$ and $C_{i,k}^{(t)}$ be the guessed $\Lambda$ value and the final cumulated completion time of machine $i$ by jobs from level $k$ in phase $t$. Use $T_i$ to represent the final completion time of machine $i$. We have 
    \begin{align*}
        \obj(\vX) &= \left(\sum_{i\in \cM} T_i^q \right)^{\frac{1}{q}} \\
        & = \left(\sum_{i\in \cM} \left(\sum_{k,t} C_{i,k}^{(t)} \right)^q \right)^{\frac{1}{q}} \\
        & \leq \sum_{k\in [K]}\sum_t \left(\sum_{i\in \cM} \left( C_{i,k}^{(t)} \right)^q \right)^{\frac{1}{q}}\tag{Minkowski inequality}~ \\
        & = \sum_t \left(\sum_{i\in \cM} \left( C_{i,K}^{(t)} \right)^q \right)^{\frac{1}{q}} + \sum_{k\in [K-1]}\sum_t \left(\sum_{i\in \cM} \left( C_{i,k}^{(t)} \right)^q \right)^{\frac{1}{q}} \\
        & \leq \bopt + \sum_{k\in [K-1]}\sum_t \left(\sum_{i\in \cM} \left( C_{i,k}^{(t)} \right)^q \right)^{\frac{1}{q}}~,
    \end{align*}
    where the last inequality follows the fact that we assign each job in level $K$ to all machines proportionally to the $\gamma$-th power of their speeds, achieving a lower bound on the optimal solution.
    
    Proving the theorem is sufficient to show that for each $k\leq K-1$, $\sum_t \left(\sum_{i\in \cM} \left( C_{i,k}^{(t)} \right)^q \right)^{\frac{1}{q}}$ is $O(1)\cdot \opt$. Similar to the proof of~\cref{lem:makespan:fractionalcompetitive}, we can discuss separately the last job in a phase and the other jobs. Applying the Minkowski inequality again, we have for each phase $t$, 
        \[\left(\sum_{i\in \cM} \left( C_{i,k}^{(t)} \right)^q \right)^{\frac{1}{q}} \leq \Lambda^{(t)} + \Lambda^{(t+1)} .\]

    As $\Lambda^{(t)}$ increases exponentially and the final $\Lambda$ is at most $2\cdot \bopt$, we have $\obj(\vX) \leq  O(\log m) \cdot \opt$.
\end{proof}

\begin{proof}[Proof of~\cref{thm:lq_ratio}]
    The theorem can be proved directly by the above lemmas.
    In \cref{lem:lq:rounding} and \cref{lem:lq:roundable}, we prove that the returned fractional allocation is $O(\log m)$-competitive and speed-size-feasible.  
    Thus, independent rounding returns a $O(m^{\frac{1}{q}(1-\frac{1}{q})} \cdot (\log m)^{1+\frac{1}{q^2}} )$-competitive solution with probability at least $1-1/m^2$. As the rounding is always $O(m^2)$-competitive, the claimed competitive expectation can be obtained.
\end{proof}
\subsection{Proof of \Cref{thm:lq}}
Finally, we conclude the proof of \Cref{thm:lq}. The same as the makespan version, we use independent rounding to round the fractional algorithm in \Cref{alg:lq}. We preserve the two-sided monotonicity in expectation, so the randomized algorithm is two-sided implementable in expectation. The competitive ratio has already been proved in \Cref{thm:lq_ratio}. For the running time, the randomized algorithm runs in polynomial time for the same reason as the makespan version. 



\section{Algorithms That Seem Truthful but Aren't}\label{sec:ce}

As we mentioned before, being well-behaved, or almost well-behaved, is an important intuition for achieving machine-side monotonicity. However, it is worth pointing out that maintaining only a well-behavior property is not sufficient. 

In this section, we present two natural online algorithms that seem to be machine-side truthful intuitively but are found to be not truthful.

\subsection{~\citet{DBLP:journals/scheduling/LiLW23}'s Well-Behaved Mechanism}
The first mechanism we discuss is given by~\citet{DBLP:journals/scheduling/LiLW23}, which is job-side truthful and machine-side almost well-behaved. The well-behavior property is originally used to describe the property of a machine with a higher speed that always has more load than a slower machine. In~\citet{DBLP:journals/scheduling/LiLW23}'s work, they focus on a stronger well-behavior property with makespan, where the faster machine always has no smaller makespan. They prove their mechanism achieves an makespan almost well-behavior property. This property only requires a machine $i$ to have no smaller makespan than the machine $i'$ only if $s_i\geq 2s_{i'}$, which is a relaxation of the makespan well-behavior property. They prove the algorithm achieves machine-side truthfulness under the case $p_j=1,\forall j\geq 1$, and case $m=2$. Therefore, it is a good candidate to achieve machine-side truthfulness. We present their algorithm briefly as follows. 

\begin{enumerate}[left=2em]
    \item Round the speed $s_i$ submitted by machine down to the largest power of $2$. That is, if $s_i\in [2^k,2^{k+1})$, then announce $s_i=2^k$.
    \item Rank the machine as $s_1\geq s_2\geq \ldots\geq s_m$. For the machines with the same speed, break tie according to the makespan, which means the machine with a larger makespan has a smaller rank.
    \item Set price $\rho_i=\sum_{i'\leq i}\pi_{i'}$ for machine $i$ before job $j$ comes, where $\pi_i=\frac{s_i}{s_{i-1}}(C_{i-1}-C_i)$ and $C_i$ is the makespan of machine $i$ before job $j$ comes. 
\end{enumerate}


Their algorithm satisfies the following two basic properties. 
\begin{enumerate}[left=2em]
    \item If scheduling job $j$ on machine $i$ breaks the well-behavior property (w.r.t the makespan), then for job $j$, the cost of machine $i-1$ is cheaper than that of machine $i$.
    \item If scheduling job $j$ on machine $i$
results in a makespan for machine $i$ that is strictly less than the makespan of machine $i-1$, then the cost for job $j$ to be scheduled on machine $i$ must be strictly lower than that on machine $i-1$. Consequently, job $j$ will not choose machine $i-1$.
    
\end{enumerate}

These properties hold since $$C_{i-1}+\frac{p_j}{s_{i-1}}+\rho_{i-1}\geq C_{i}+\frac{p_j}{s_{i}}+\rho_{i}\Leftrightarrow C_{i-1}\geq C_{i}+\frac{p_j}{s_{i}}$$.



Their algorithm is shown to be truthful in special case $p_j=1$ and case $m=2$. It is worth noting that to prove the truthfulness, they need to round the speed to the power of $4$ other than $2$ for the case $m=2$. 
It might seem that we can use larger round factor $a'$ to maintain machine-side truthfulness when $m$ becomes large. However, we construct a hard instance to show that their approach cannot achieve machine-side truthfulness in the case $m=3$ no matter how large $a'$ is. 




Assuming that the algorithm rounds down the machine speed to an integer power of $a'\geq 1$, we fix a parameter $a \geq 2$ that $a$ is an integer power of $a'$. 
We have three machines. 
$$
s_1=a^{x+1},~s_2=a^{x},~s_3=a^0=1.
$$
Here, $x$ is an arbitrary integer such that $x \geq 2$. 
The instance has six jobs arriving in the following order.
$$
p_1=a^{x+1},
~p_2=a^{x}-\epsilon,
~p_3=1,
~p_4=a^{x+1}\cdot k,
~p_5=a^{x} \cdot k + \frac{1}{a} + \frac{\epsilon}{2}
,~p_6=\frac{k}{2} 
~.
$$
Here, $k$ is super large and $\epsilon$ is super small, we may have some different constraints on $k$ and $\epsilon$, which will be all satisfied when $k$ is big enough 
 and $\epsilon$ is small enough. 

Their mechanism works as follows. (Refer to the allocation in \cref{fig:wb-ce1}.) 
For the first job $p_1$, the mechanism will always allocate it on $s_1$ since all the machines have makespan $0$. For the second job $p_2$, the makespan of machine $2$ after scheduling $p_2$ is strictly smaller than the makespan of machine $1$ since the small $\epsilon$. If $p_2$ schedules on machine $3$, it breaks the well-behavior property w.r.t. makespan. Hence $p_2$ must schedule on machine $2$. Then, for job $p_3$, when $\epsilon$ is small enough, scheduling $p_3$ on either machine $2$ or machine $3$ will break the well-behavior property. We will schedule it on machine $1$. When $p_4$ comes, by setting a large enough $k$, machine $1$ is still the only option to keep the well-behavior property, so we schedule $p_4$ on machine $1$. When $p_5$ arrives, again because $k$ is large enough, scheduling on machine $3$ is impossible. On the other hand, if we schedule on machine $2$, the makespan will be $k + \frac{1}{a^{x+1}} - \frac{\epsilon}{2a^{x}}$ still slightly smaller than the {makespan} of machine $1$, which is $k + \frac{1}{a^{x+1}}$. By the second property, we cannot schedule it on machine $1$, so machine $2$ is the only choice. For the final job $p_6$, if we schedule it on machine $3$, the makespan will be $k/2$, smaller than the makespan of machine $2$, so again according to the second property, we may schedule it on machine $1$ or machine $3$. Let us calculate their costs.
For machine $1$, the cost is 
$$
C_{1} + \frac{p_6}{s_1} + \rho_1=C_{1} + \frac{p_6}{s_1}=\frac{1}{a^{x+1}}\cdot (p_1+p_3+p_4 + p_6) = k + \frac{k}{2a^{x+1}} + \frac{1}{a^{x+1}} + 1 = k + \frac{k}{2a^{x+1}} + O(1)~.
$$
For machine $3$, the cost is
$$
C_{3} + \frac{p_6}{s_3} + \rho_3=C_{3} + \frac{p_6}{s_3}+\frac{s_2}{s_1}(C_1-C_2)+\frac{s_3}{s_2}(C_2-C_3)=k/2 + \frac{1}{a^x} \cdot \frac{p_2 + p_5}{a^x} = k/2 + \frac{k}{a^x} + O(1)~.
$$
When $a \geq 2$ and $x \geq 2$, it is straightforward to see the cost to schedule on machine $1$ is more expensive, so $p_6$ will choose machine $3$. 


Now, considering that machine $3$ speeds up to $a$, their mechanism works as follows. (Refer to the allocation in \cref{fig:wb-ce2}.) 
$p_1$ and $p_2$ will choose machine $1$ and machine $2$, respectively. However, the difference happens on $p_3$. We only know that we will not schedule $p_3$ on machine $2$ when $\epsilon$ is small enough according to property 1. 
Then we argue that $p_3$ will be scheduled on machine $3$ due to lower cost. The cost of $p_3$ on machine $1$ is 
$$
C_{1} + \frac{p_3}{s_1} + \rho_1=C_{1} + \frac{p_3}{s_1}=1+\frac{1}{a^{x+1}}
$$ 
and that on machine $3$ is
$$
C_{3} + \frac{p_3}{s_3} + \rho_3=\frac{p_3}{s_3}+\frac{s_2}{s_1}(C_1-C_2)+\frac{s_3}{s_2}(C_2-C_3)=\frac{1}{a}+\frac{1}{a}\cdot\frac{\epsilon}{a^x}+\frac{1}{a^{x-1}}\cdot(1-\frac{\epsilon}{a^x})~.
$$
It is straightforward to verify that when $\epsilon$ is sufficiently small, machine~$3$ is the cheaper option. Hence, job $p_3$ is assigned to machine~$3$. Subsequently, job $p_4$ is still scheduled on machine~$1$, since the parameter $k$ is large enough that assigning $p_4$ to either $s_2$ or $s_3$ would violate the well-behavior property.

However, the situation changes upon the arrival of job $p_5$. Since machine~$1$ no longer processes job $p_3$, assigning $p_5$ to either $s_2$ or $s_3$ would violate the well-behavior property with respect to the makespan. Consequently, $p_5$ must be scheduled on machine~$1$. Moreover, this change alters the subsequent assignment: the final job cannot be placed on machine~$3$, as doing so would again violate the well-behavior property.

Finally, machine $3$ gets $p_6 = k/2$ before speeding up, and gets $p_3=1\ll k/2$ after speeding up. Therefore, the load of machine $3$ is not monotone when it speeds up. By \Cref{lem:machine-side-implementable}, their mechanism cannot be implemented by a machine-side truthful algorithm with payments.

\begin{figure}[htbp]
    \centering
    \begin{subfigure}
        {0.49\linewidth}
		\includegraphics[width=\textwidth]{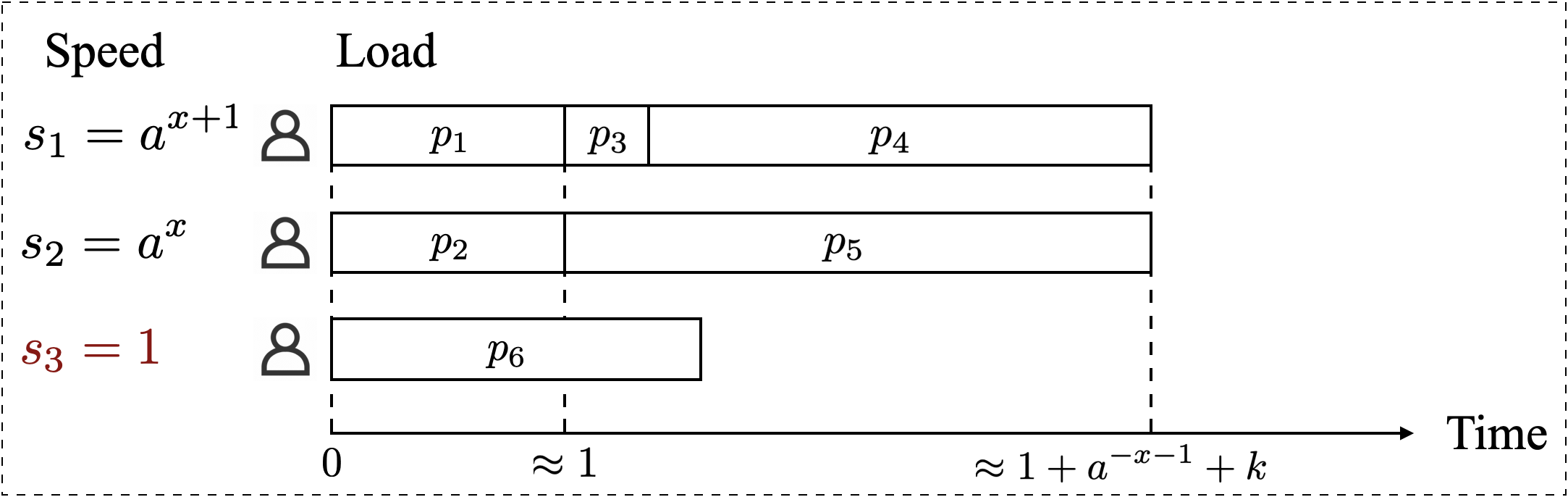}
		\caption{Before Speeding Up}
\label{fig:wb-ce1}
	\end{subfigure}
 \begin{subfigure}{0.49\linewidth}
		\includegraphics[width=\textwidth]{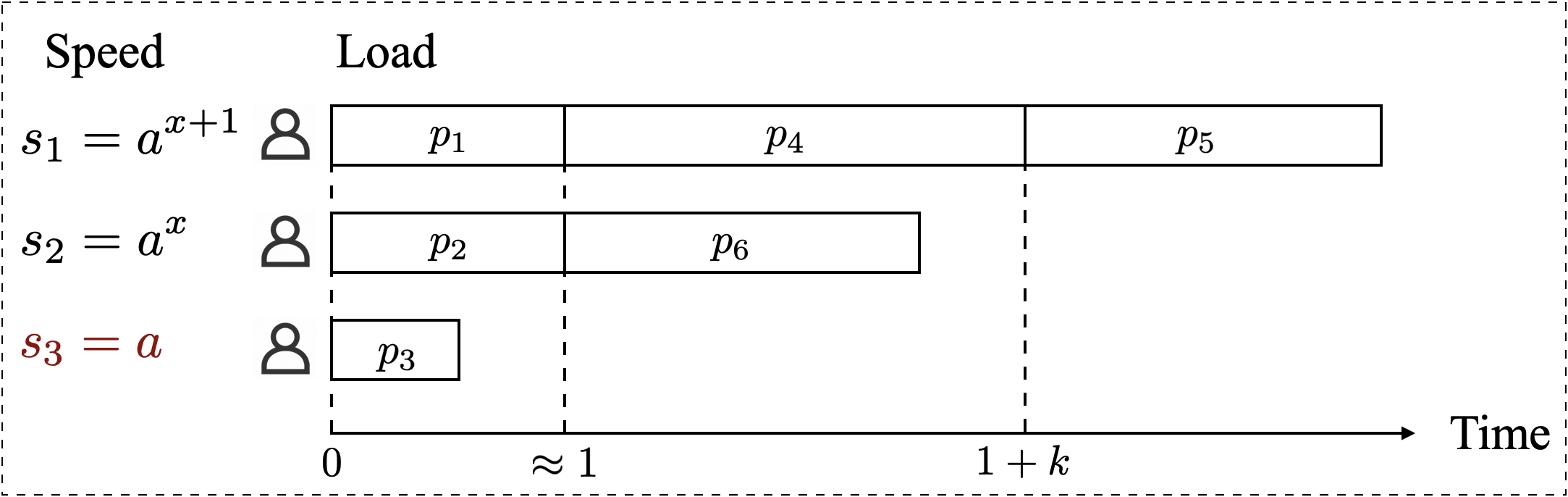}
		\caption{After Speeding Up}
\label{fig:wb-ce2}
	\end{subfigure}
 \caption{A counter-example shows that ~\citet{DBLP:journals/scheduling/LiLW23}'s mechanism is not truthful.}
\end{figure}

\subsection{Water-Filling Algorithm}
We observe that the key reason that makes \citet{DBLP:journals/scheduling/LiLW23}'s algorithm not truthful in the hard instance is that some jobs only break the well-behavior property a little but make the schedule totally different. (e.g., $p_5$ in the example, it is very large but totally moved from machine $2$ to machine $1$.) It will be circumvented when we consider a fractional algorithm, and only maintaining the fractionally well-behavior property. (e.g., we still schedule some part of $p_5$ on machine $2$.) After that, we can avoid such an intense difference between the two schedules because of the schedule of only one job. This idea is a key component of our two-sided truthful mechanism in \Cref{alg:makespan}. 

The question is, after using a fractional algorithm, is the fractionally well-behavior property sufficient for machine-side truthfulness? Unfortunately, the answer is "no".
We found that many natural fractionally well-behaved algorithms we have tried do not have machine-side monotonicity.

Consider the following \Cref{alg:wtf}. The algorithm also has a doubling procedure with a guessed \opt ($\Lambda$), and also a speed-size-feasible\footnote{We can remove all machines with speed at most $s_1/m$ first to promise the second property of speed-size-feasible, for convenience of presenting the counter-example, we do not include this operation here.} algorithm as \Cref{alg:makespan}. Thus, it also promises a nice competitive ratio after independent rounding. 
When each online job arrives, the allocation follows the idea of water-filling, where the water level of each machine is its makespan. We continuously match some portion of $p_j$ to the feasible machine with the lowest makespan, where feasible means the machine's speed satisfies the speed constraint w.r.t. $p_j$ and $\Lambda$, i.e., $s_i\Lambda\geq p_j$. The algorithm is fractionally well-behaved and seems to satisfy machine-side monotonicity because when we speed up a machine, it can join more job's water-filling procedures. However, we will give a hard instance to show it does not have monotonicity. The main reason is that we have a doubling procedure; when $\Lambda$ becomes different for the same job, the monotonicity may not hold. Compared to the proportional algorithm in \Cref{alg:makespan}, the schedule of jobs in \Cref{alg:wtf} may not be as stable as that in \Cref{alg:makespan}, which makes us lose the machine-side monotonicity.

\begin{algorithm}[htbp]
    \caption{Water-Filling}
    \label{alg:wtf}
    \KwIn{The reported machine speeds $\{s_i\}_{i\in [m]}$, the reported job sizes $\{p_j\}_{j\in [n]}$ which shows up online, and the parameter $q \geq 1$. }
    \KwOut{ A fractional online allocation $\{x_{ij}\}_{i\in [m],j\in [n]}$ }
    sort $s_i$ in descending order \;
    Assign $p_1$ equally on machines whose speeds are equal to $s_1$\;
    $\Lambda \gets \frac{p_1}{s_1}$  \tcp*{initialize a guessed optimal objective}
    \For{each arriving job $p_j$ }
    {
    $x_j \gets 0$\;
    Initialize $x_{ij} \gets 0$ for each $i\in [m]$\;
    \While{$x_j< 1$}{
    \If{$s_1 \cdot \Lambda< p_j$}{Keep doubling $\Lambda$ until $s_1 \cdot \Lambda\geq p_j$\;}
    $v \gets \arg \min_{i}\{\sum_{j'<j}\frac{x_{ij'}\cdot p_{j'}}{s_i}\mid s_i\Lambda\geq p_j,~ i\in [m]\}$\;
    $x_{vj} \gets x_{vj} + \mathrm{d}x$ \tcp*{$\mathrm{d}x$ is a small enough portion}
    $x_j \gets x_j + \mathrm{d}x$\;
    \If{$\sum_{j'<j}\frac{x_{vj'}\cdot p_{j'}}{s_v}=\Lambda$}{
    $\Lambda=2\cdot \Lambda$\;
    }
    }
    }
    \Return{$\{x_i\}_{i\in [m]}$.}
\end{algorithm}

We have five machines in the hard instance.\footnote{If we have the operation of removing machines, we need to create some dummy machines to be removed.}
$$
s_1=64,~s_2=32,~s_3=16,~s_4=4,~s_5=2.
$$ 
We have six jobs arriving in the following order.
$$
p_1=64,~p_2=\epsilon,~p_3=32-\frac{32}{54}\epsilon+\delta,~p_4=p_5=8-\frac{9}{56}\epsilon,~p_6=\frac{8}{5}-\frac{9}{5\cdot 56}\epsilon.
$$ 

Here, $\delta$ and $\epsilon$ are both small positives such that $0<\delta\ll\epsilon$. 

The mechanism works as follows. (Refer to the allocation in \cref{fig:wtf-ce1}.) When the first job $p_1$ comes, the mechanism allocates it on machine $1$ and initializes $\Lambda=1$. For the second job $p_2$, all the machines are feasible. As a result, it's allocated on machines $2$ to $5$ proportionally and raises the water level to $\epsilon\cdot 1/54$. For job $p_3$, only machine $1$ and machine $2$ are feasible, so it will first raise machine $2$'s water level to $1$, then find that both machine $1$ and machine $2$ are full but still $\delta$ fraction unassigned. Therefore, the mechanism will double $\Lambda$ to $2$, and assign $\delta$ to feasible machines $1$, $2$, and $3$. However, it's so small that we ignored it when we analyzed the water level below. The jobs $p_4$ and $p_5$ come, they are feasible for machines $1$ to $4$. Therefore, they will be assigned to machines $3$ and $4$ proportionally, raising their water level to about $4/5$. Lastly, the job $p_6$ comes. It is feasible on all the machines, so it will be allocated on machine 5 and raise the water level to about $4/5$, which is the same as that of machines $3$ and $4$.

Now, considering that machine $5$ speeds up to $4$, the mechanism works as follows. (Refer to the allocation in \cref{fig:wtf-ce2}.) The jobs $p_1$ and $p_2$ are assigned following the same strategy, but it should be noted that the water levels of machines $2$ to $4$ are only raised to $\epsilon\cdot 1/56$. Therefore, when job $p_3$ comes, it will raise machine $2$'s water level to almost $1$ but doesn't reach $1$, as a result of which, it does not cause $\Lambda$ double. Therefore, for the jobs $p_4$ and $p_5$, only machines $1$ to $3$ are feasible. They are allocated on machine $3$ to raise the water level to about $1$, but also, do not reach $1$ and do not cause $\Lambda$ double. Then, when $p_6$ comes, it will be allocated on machines $4$ and $5$ proportionally, raising their water level to $1/5$.

Finally, we observed that the load of machine $5$ is not monotone when it speeds up. The load before is about $\frac{8}{5}$, and the load after is about $\frac{4}{5}$. It contradicts the machine-side monotonicity, thus not machine-side implementable by \Cref{lem:machine-side-implementable}.
The main difference during the process between the before speeding up and after speeding up setting is that when $p_3$ comes, it will cause $\Lambda$ double before $s_5$ speeding up while it won't after $s_5$ speeding up. As a result, after $s_5$ speeds up, machine $4$ can not get jobs $p_4$ and $p_5$, due to which it takes away parts of $p_6$ from machine $5$. Therefore, fewer loads are allocated to machine $5$ after it speeds up. Notice that all the machines' speed is a power of $2$, so there is no difference whether we introduce a round-down procedure as \Cref{alg:makespan}.

\begin{figure}[htbp]
    \centering
    \begin{subfigure}{0.49\linewidth}
		\includegraphics[width=\textwidth]{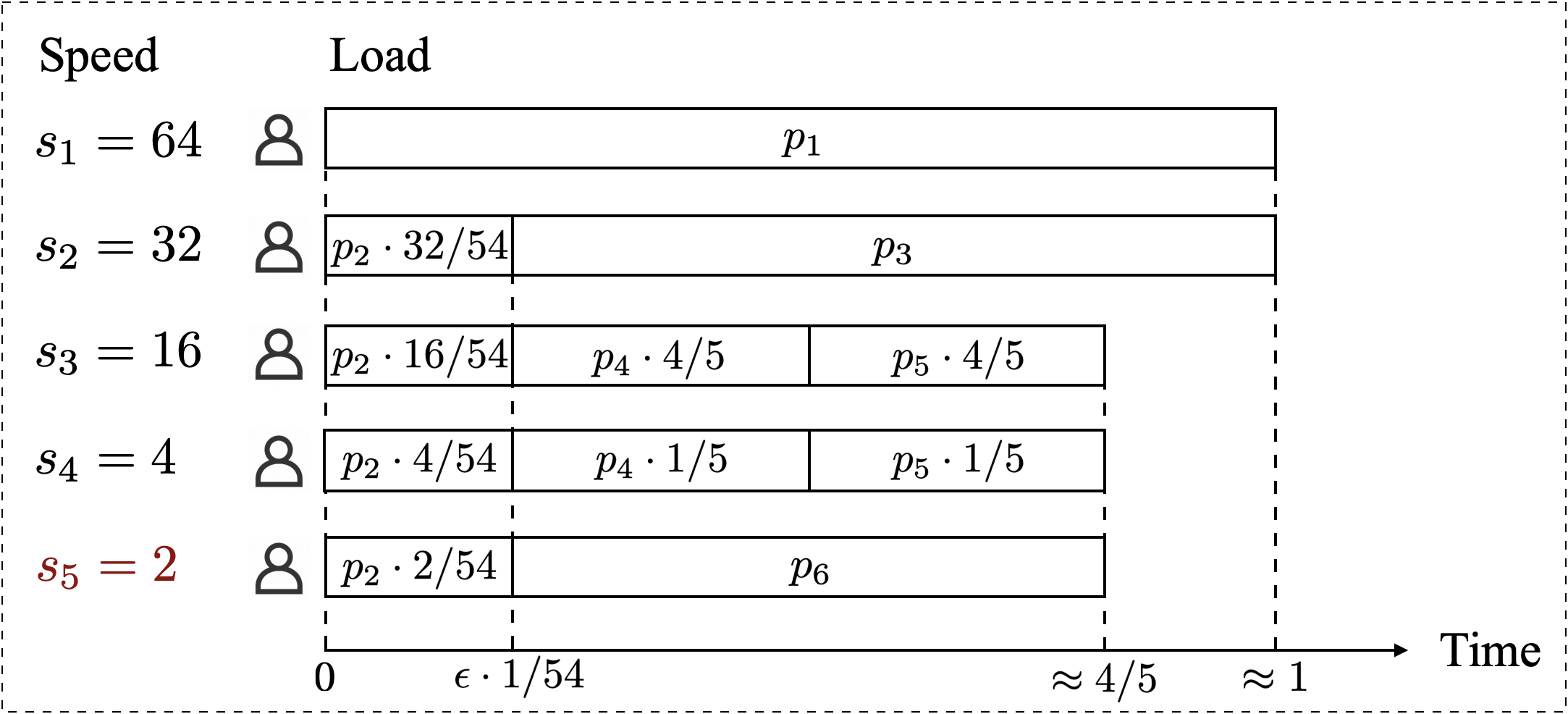}
		\caption{Before Speeding Up}
  \label{fig:wtf-ce1}
	\end{subfigure}
 \begin{subfigure}{0.49\linewidth}
		\includegraphics[width=\textwidth]{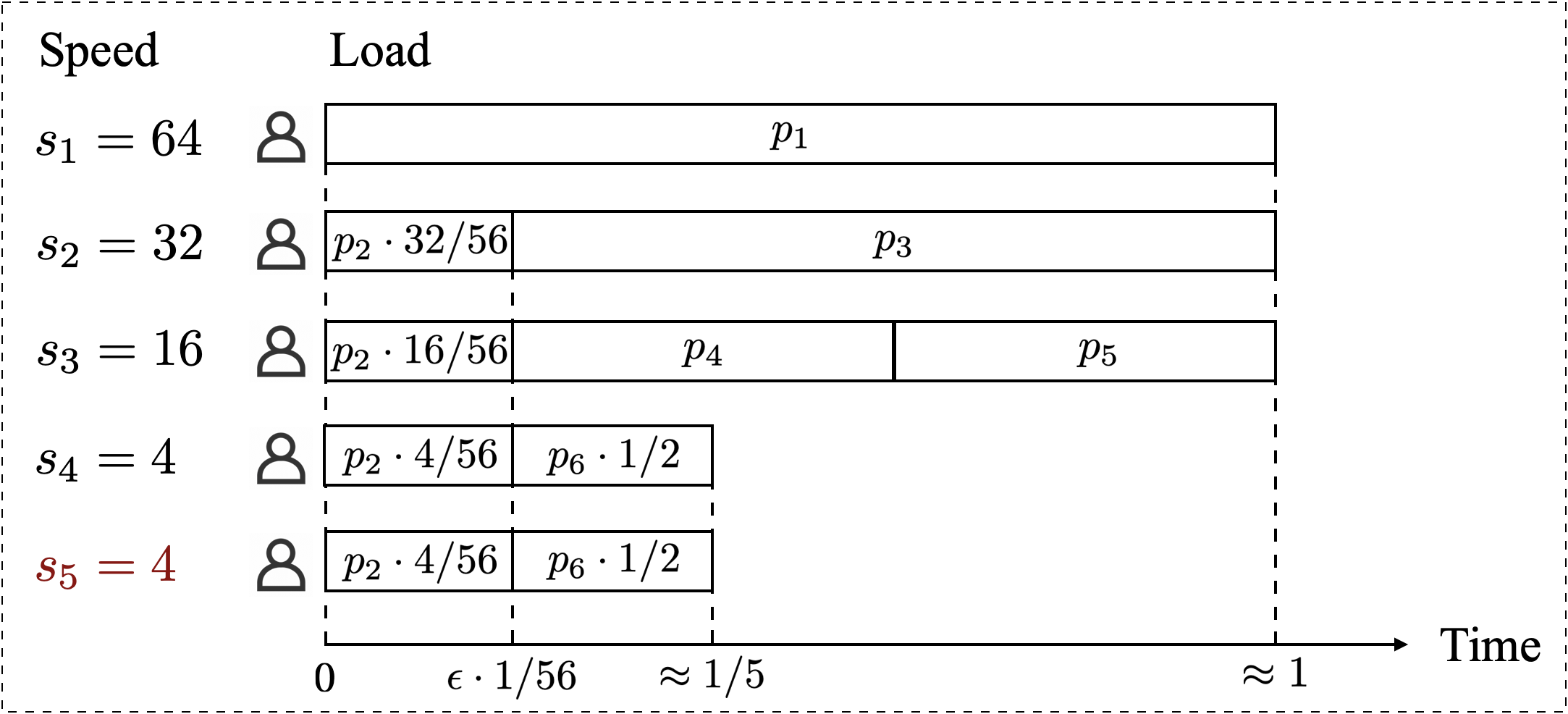}
		\caption{After Speeding Up}
  \label{fig:wtf-ce2}
	\end{subfigure}
 \caption{A counter-example shows that ~\cref{alg:wtf} is not truthful.}

\end{figure}

\section{Proof of \Cref{lem:job-side-implementable}}\label{sec:job-truth}


Consider a randomized algorithm; for each arriving job $j$, let $x_{ij}(p)$ be the probability that mechanism $\cA$ assigns it to machine $i$ when the job reports size $p$. We use a randomized mechanism $\cA$, which charges a fee to a scheduled job by a specific payment function $Q_j(p)$, to implement the randomized algorithm. We have the following lemma, which implies \Cref{lem:job-side-implementable}.


\begin{lemma}\label{lem:job_truth}
The mechanism $\cA$ is truthful on job side if and only if the following two properties are satisfied:
\begin{itemize}
    \item Job-Side-Monotone: For each job $j$, the expected unit processing time $ \sum_{i\in [m]}x_{ij}(p)\cdot \frac{1}{s_i}$ is non-increasing as $p$ increases.
    \item For each job $j$, given any current makespan $C_i^{(j)}$ on machine $i$, the expected fee is given by a payment function $Q_j(p) = Q_j(0) - \sum_{i\in [m]} \left(C_i^{(j)}\cdot \left(x_{ij}(p)-x_{ij} (0)\right) + \frac{1}{s_i} \cdot \left(p \cdot x_{ij}(p) - \int_{0}^{p}x_{ij}(t) \mathrm{d}t \right)\right) $, where $p$ is its submitted size. \footnote{The reason we need $Q_j(0)$ is to guarantee every $Q_j(p)\geq 0$ if it is required.}
\end{itemize}
\end{lemma}

\begin{proof}
First, we prove that the two properties above are necessary conditions of job-side truthfulness. Consider a job $j$ whose true job size is $t$. The expected cost of job $j$ submitting its size $p$ is 
$$
    C(p)=\sum_{i\in [m]} x_{ij}(p)[C_i^{(j)}+t\cdot\frac{1}{s_i}]+Q_j(p)~.
$$
$C(p)$ is minimum at $p=t$ for all $t$ only if the first derivative of $C(p)$ at $p=t$ is zero and the second derivative of $C(p)$ at $p=t$ is non-negative. According to the first derivative property, 
\begin{align*}
    \left.\left[\sum_{i\in [m]} \frac{\mathrm{d}x_{ij}(p)}{\mathrm{d}p}[C_i^{(j)}+t\cdot\frac{1}{s_i}]+\frac{\mathrm{d}Q_j(p)}{\mathrm{d}p}\right]\right|_{p=t}=0,\quad \forall t~.
\end{align*}
Integrating each part, we have 
\begin{align}\label{eq:job-mon2}
    Q_j(p) = Q_j(0) - \sum_{i\in [m]} \left(C_i^{(j)}\cdot \left(x_{ij}(p)-x_{ij} (0)\right) + \frac{1}{s_i} \cdot \left(p \cdot x_{ij}(p) - \int_{0}^{p}x_{ij}(t) \mathrm{d}t \right)\right)
\end{align}
According to the second derivative property, 
\begin{align*}
    \left.\left[\sum_{i\in [m]} \frac{\mathrm{d}^2x_{ij}(p)}{\mathrm{d}p^2}[C_i^{(j)}+t\cdot\frac{1}{s_i}]+\frac{\mathrm{d}^2Q_j(p)}{\mathrm{d}p^2}\right]\right|_{p=t}\geq 0,\quad \forall t~.
\end{align*}
Substituting \cref{eq:job-mon2} into the above equation, we obtain
\begin{align}\label{eq:job-mon3}
    \sum_{i\in [m]}x'_{ij}(p)\cdot \frac{1}{s_i}\leq 0,\quad \forall p
\end{align}
which means $\sum_{i\in [m]}x_{ij}(p)\cdot \frac{1}{s_i}$ is non-increasing as $p$ increases. It should be noted that $\sum_{i\in [m]}x_{ij}(p)\cdot \frac{1}{s_i}$ means the average processing time a job submitted size $p$ has.
This equation shows the average processing time should not increase if a job submits a larger size.

Therefore, the mechanism is job-side truthful only if the pricing function follows \cref{eq:job-mon2} and the allocation probability function $x_{ij}$ satisfies \cref{eq:job-mon3}.

Next, we prove that if \cref{eq:job-mon3} is satisfied, the mechanism is job-side truthful with the pricing strategy as shown in \cref{eq:job-mon2}.
Consider any $a,b>0$, and without loss of generality, assume that $a<b$. If the job's true size is $a$ and submit the size $b$, it will pay 
$$\int_{a}^{b}\sum_{i\in [m]}x_{ij}(u)\cdot \frac{1}{s_i}\mathrm{d}u-\sum_{i\in [m]}x_{ij}(b)\cdot \frac{1}{s_i}\cdot[b-a]$$ more than submitting its true size $t$ honestly, which is non-negative when $\sum_{i}x_{ij}(p)\cdot \frac{1}{s_i}$ is non-increasing as $p$ increases.
Symmetrically, if the job's true size is $b$ and submit the size $a$, it will pay 
$$\sum_{i\in [m]}x_{ij}(a)\cdot \frac{1}{s_i}\cdot[b-a]-\int_{a}^{b}\sum_{i\in [m]}x_{ij}(u)\cdot \frac{1}{s_i}\mathrm{d}u$$ more than submitting its true size $t$ honestly, which is also non-negative when $\sum_{i\in [m]}x_{ij}(p)\cdot \frac{1}{s_i}$ is non-increasing as $p$ increases.

In conclusion, the two properties in \cref{lem:job_truth} are both sufficient and necessary conditions for job-side truthful mechanism design.
Additionally, if $Q_j(p) \geq 0$ is required, we need to set a large enough $Q_j(0)$ unrelated to $p$. 
Since $\sum_{i\in [m]}x_{ij}(p)\cdot \frac{1}{s_i}$ is non-increasing, $$p\cdot\sum_{i\in [m]}  \frac{1}{s_i}\cdot x_{ij}(p) \leq \int_{0}^{p}\sum_{i\in [m]}  \frac{1}{s_i}\cdot x_{ij}(t) \mathrm{d}t~.$$ By setting $Q_j(0)=\sum_{i\in [m]} C_i^{(j)}$, we have $$Q_j(p)\geq  Q_j(0) - \sum_{i\in [m]} C_i^{(j)}\cdot \left(x_{ij}(p)-x_{ij} (0)\right) \geq  Q_j(0) - \sum_{i\in [m]} C_i^{(j)}= 0~.$$
\end{proof}


\section{Discussion of Our Doubling Operations}
\subsection{Counterexample for Allocate-Before-Doubling}\label{sec:abd}
In this section, we show that if we allocate jobs after doubling the $\Lambda$ value, there exists a counterexample breaking job-side monotonicity.
Assume that there are $m=8$ machines submitting their speeds as  $s_1=8$, $s_2=4$, and $s_3=s_4=\ldots=s_8=2$. We have $K=\lfloor \log m\rfloor +1=4$ groups of machines, where $\cM_1=\{1\},\cM_2=\{2\},\cM_3=\{3,4,\ldots,8\},\cM_4=\emptyset$.

When the first job $p_1=8$ arrives,  we place it on machine $1$ and set $\Lambda=1$. Then three jobs with size $3$ arrive. Since $3\in ( r_3\cdot \Lambda, r_2\cdot \Lambda]=(2,4]$, these jobs' levels are $2$. Then, the mechanism will place these jobs on machine $1$ and $2$ proportionally, resulting in $C_{1,2}=C_{2,2}=3/4$. 

Now we discuss the fifth job whose real size is $3$. If it reports its size truthfully, it will be a level $2$ job, and the unit processing time is $1/6$. However, if it misreports a larger size of $3+\epsilon$, the mechanism will double $\Lambda$ to $2$ first and allocate it then. Notice that after $\Lambda$'s doubling, $3+\epsilon \in ( r_4\cdot \Lambda, r_3\cdot \Lambda]=(2,4]$. Therefore, this job's level becomes 3.
This breaks the job-side monotonicity, as its unit processing time increases to $1/3$.

\subsection{Counterexample for Double-Without-The-Last}\label{sec:dwtl}
In this section, we show that if we double $\Lambda$ with the last level, the $\Lambda$-stability (refer to \cref{lem:makespan:lambda_stable}) breaks.
Consider $m=18$ machines where $\bs_1=16$, $\bs_2=4$, $\bs_3=\ldots=\bs_{18}=2$ ($16$ machines with speed $2$). We will have $K=\lfloor\log m\rfloor+1=5$ groups of machines, where $\cM_1=\{1\},\cM_2=\emptyset,\cM_3=\{2\},\cM_4=\{3,\ldots,18\}, \cM_5=\emptyset$ with group speed $r_1=16,r_2=8,r_3=4,r_4=2,r_5=1$.


In this example, we show that when $\bs_2$ speeds up to $8$, $\Lambda$ can be four times larger than $\Lambda'$ at some time point, which breaks the $\Lambda$-stability. 
For analysis convenience, we show the doubling process in a \emph{phase-budget} form. That is, the algorithm is viewed as several phases with different values of $\Lambda$. In each phase, there are $K=\lfloor \log m \rfloor+1$ empty bins, each corresponding to a job level. The jobs in level $k$ budget are allocated proportionally to the machines in $\cpM{k}$. The budget of level $k$ bin is calculated as $B_k = \Lambda\cdot\sum_{i\in \cpM{k}}\bar{s}_i$. For instance, the budget of level $4$ bin when $\Lambda=1$ is $B_4=\Lambda\cdot(s_1+s_2+s_3+\ldots+s_{18})=52$. The $\Lambda$ doubles right after the bin of some level is overfilled, including level $K$.
\begin{figure}[htbp]
    \centering
    \begin{subfigure}{0.49\textwidth}
        \centering
        \includegraphics[width=\linewidth]{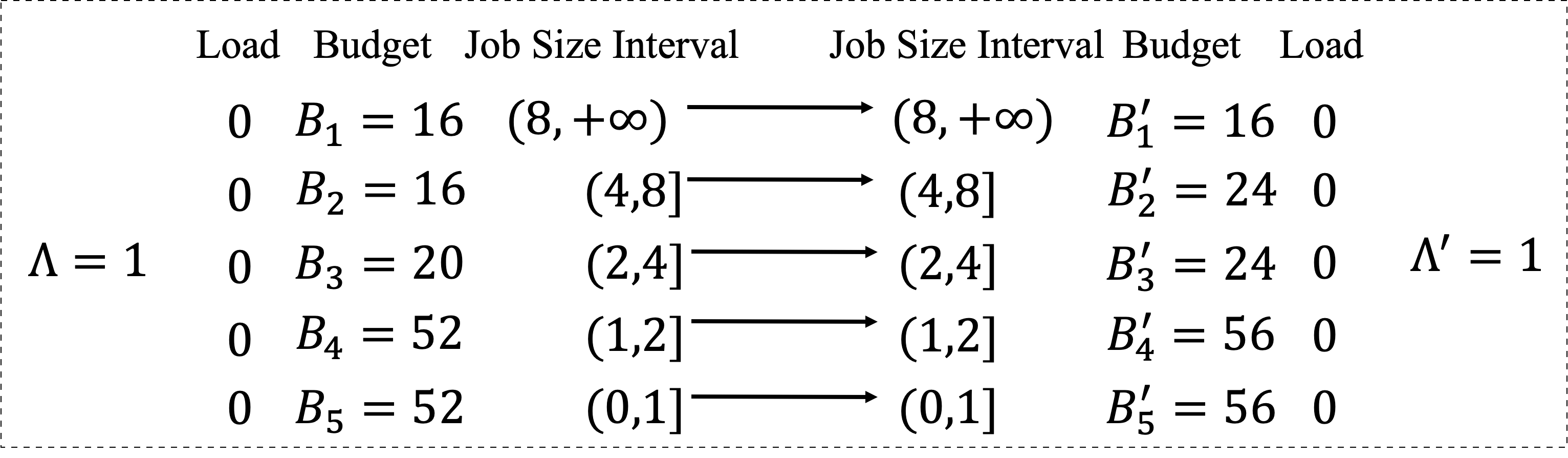}
        \caption{Initial configuration.}
        \label{fig:LS1}
    \end{subfigure}
    \begin{subfigure}{0.49\textwidth}
        \centering
        \includegraphics[width=\linewidth]{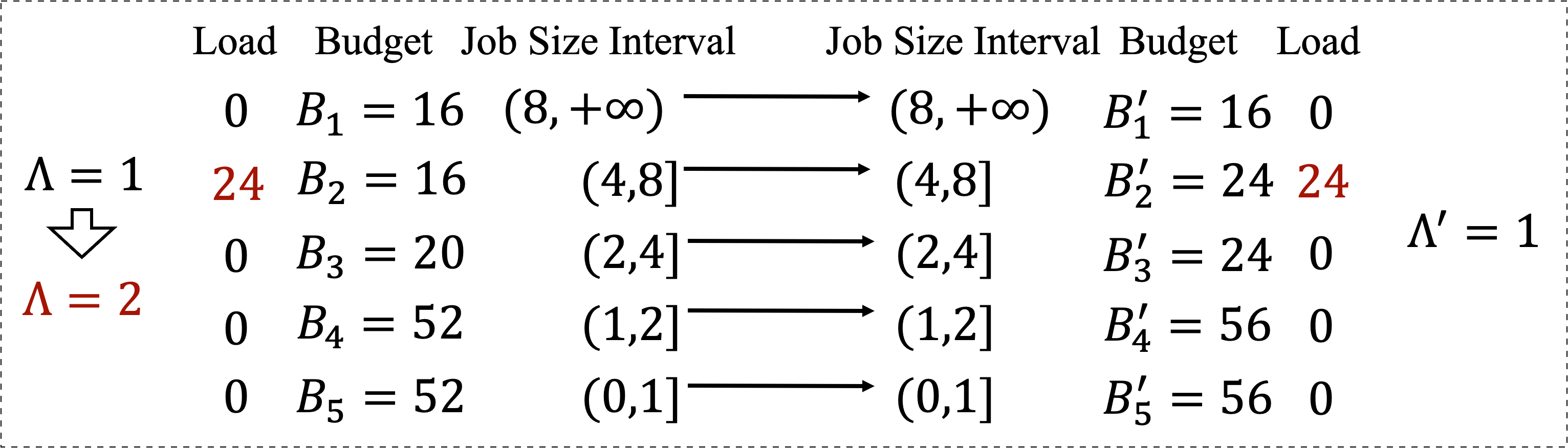}
        \caption{Allocating three job with size $8$. }
        \label{fig:LS2}
    \end{subfigure}
    \begin{subfigure}{0.49\textwidth}
        \centering
        \includegraphics[width=\linewidth]{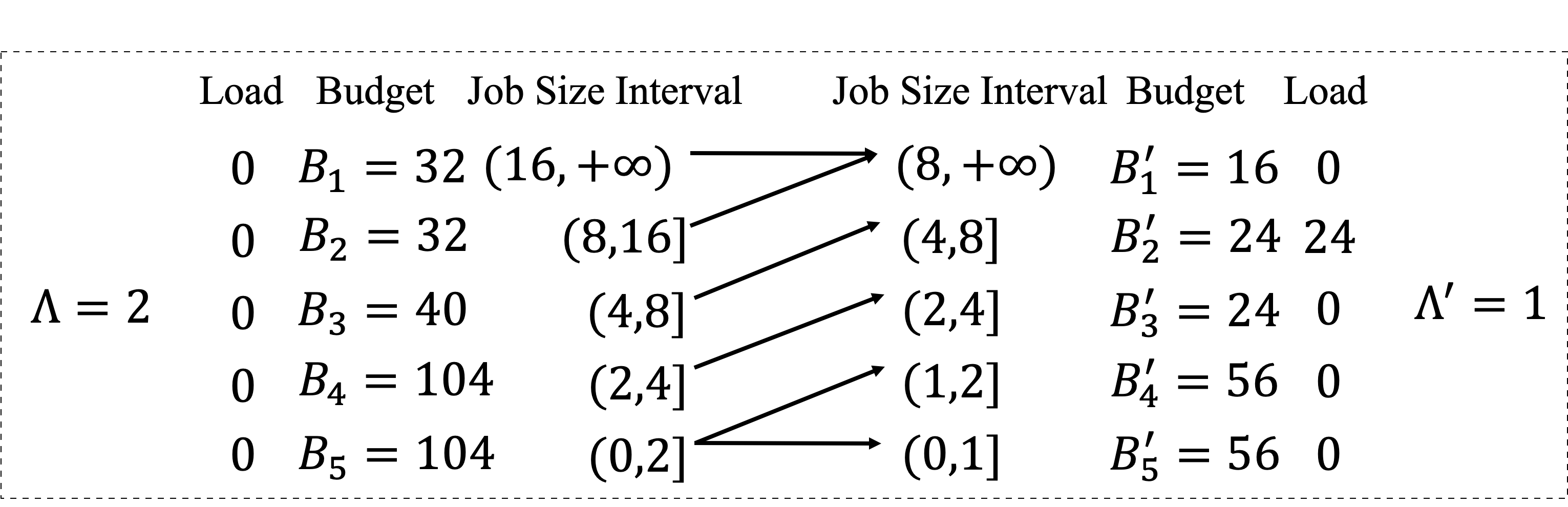}
        \caption{New configuration after \ref{fig:LS2}.}
        \label{fig:LS3}
    \end{subfigure}
    \begin{subfigure}{0.49\textwidth}
        \centering
        \includegraphics[width=\linewidth]{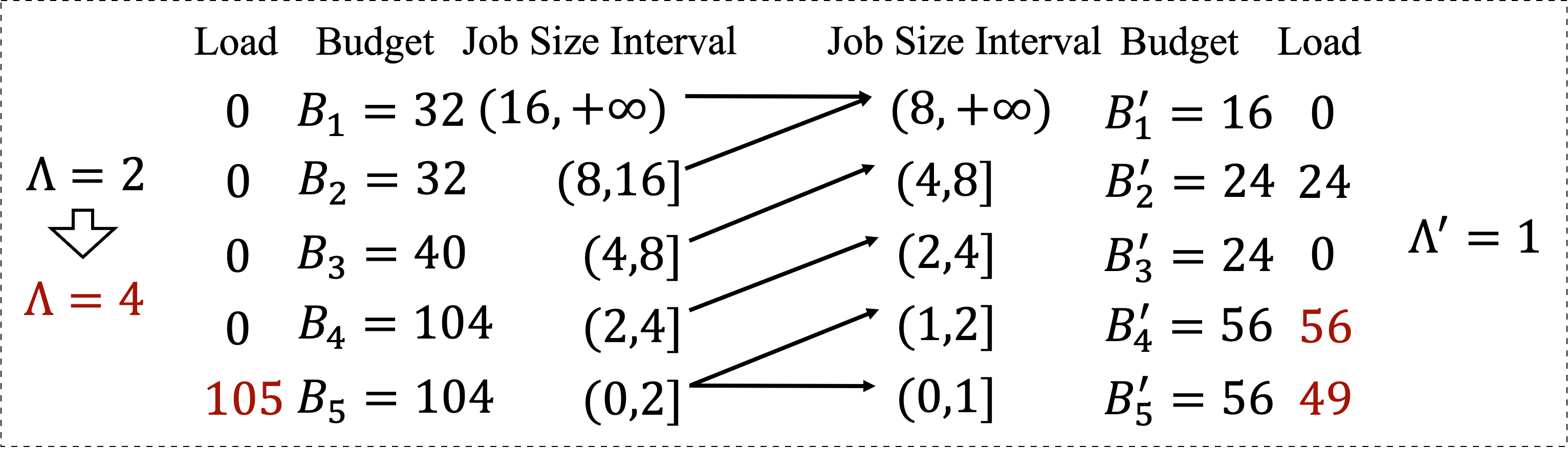}
        \caption{Allocating $28+49$ jobs with size $1$ or $2$.}
        \label{fig:LS4}
    \end{subfigure}
    \caption{After $3$ jobs with size $8$ arrive, the change of loads is shown in \cref{fig:LS2}. The new configuration is shown in \Cref{fig:LS3}. Then $28$ jobs with size $2$ and $49$ jobs with size $1$ arrive, resulting in the change shown in \cref{fig:LS4}.}
\end{figure}

The process initializes with the arrival of the first job with size $16$, where both $\Lambda$ and $\Lambda'$ are set to $1$. (Refer to the initial configuration in \Cref{fig:LS1}). Then, three jobs with size $8$ arrive. The change is shown in \cref{fig:LS2}. All of them are level $2$ jobs in both worlds. Observe that their total load is $3\times 8=24$. In the original world, they overfill $B_2=16$, while in the speed-up world, they do not overfill $B_2'=24$. Therefore, the algorithm will double $\Lambda$ but not $\Lambda'$.

After $\Lambda$ is doubled to $2$, the budget of each level's bin and the job size interval will change correspondingly in the original world; additionally, the load (or makespan) of each level will be reset to zero. The configuration after doubling is shown in \cref{fig:LS3}.


Then $28$ jobs with size $2$ and $49$ jobs with size $1$ arrive. The change is shown in \cref{fig:LS4}. In the speed-up world, since the level of jobs with size $2$ is $4$ ($2\in (r_5'\cdot\Lambda',r_4'\cdot\Lambda']=(1,2]$), the load in level $4$ becomes $28\cdot 2=56$, which doesn't overfill $B_4'=56$. Similarly, we have the level of job with size $1$ is $5$, and the load in level $5$ is $49\cdot 1=49$, which also does not overfill $B_5'=56$. Therefore, $\Lambda'$ remains to be $1$. 

However, in the original world, since all of the $28+49$ jobs (with size $2$ or $1$) are level $5$ jobs ($1,2\in (0,r_5\cdot\Lambda]=(0,2]$), the load in level $5$ becomes $28\cdot 2+49\cdot 1=105$ after all these jobs are allocated, which overfill $B_5=104$. Therefore, $\Lambda=2$ will double to $4$. We have $\Lambda'=\frac{1}{4}\Lambda$, which breaks the $\Lambda$-stability.

\end{document}